\documentclass[12pt, draftclsnofoot, onecolumn]{IEEEtran}
\normalsize
\usepackage{algorithm,amsbsy,amsmath,amssymb,epsfig,bbm,mathrsfs,multirow,amsthm}
\usepackage{color}
\usepackage{algpseudocode}
\usepackage{graphicx,caption,subcaption,array,multirow,bm}%for three figures side-by-side
\usepackage{makecell}
\usepackage{cite}
\usepackage{threeparttable}

\renewcommand{\baselinestretch}{1.32238}
\DeclareMathAlphabet{\mathpzc}{OT1}{pzc}{m}{it}

\newtheorem{observation}{Observation}
\begin{document}

\title{Defeating Super-Reactive Jammers With Deception Strategy: Modeling, Signal Detection, and Performance Analysis}

\author{ Nguyen Van Huynh, Diep N. Nguyen, Dinh Thai Hoang, Thang X. Vu,\\
	Eryk Dutkiewicz, and Symeon Chatzinotas
	\thanks{Nguyen Van Huynh, Diep N. Nguyen, Dinh Thai Hoang, and Eryk Dutkiewicz are with the School of Electrical and Data Engineering, University of Technology Sydney, NSW 2007, Australia (e-mail: huynh.nguyenvan@student.uts.edu.au; diep.nguyen@uts.edu.au; hoang.dinh@uts.edu.au; eryk.dutkiewicz@uts.edu.au).}
	\thanks{Thang X. Vu and Symeon Chatzinotas are with the Interdisciplinary Centre for Security, Reliability and Trust (SnT), University of Luxembourg, L-1855 Luxembourg (e-mail: thang.vu@uni.lu; symeon.chatzinotas@uni.lu).}
	\thanks{Preliminary results in this paper were presented at the IEEE WCNC Conference, 2020~\cite{WCNC}.}
}
\maketitle
\thispagestyle{empty}
\begin{abstract}

This paper develops a novel framework to defeat a super-reactive jammer, one of the most difficult jamming attacks to deal with in practice. Specifically, the jammer has an unlimited power budget and is equipped with the self-interference suppression capability to simultaneously attack and listen to the transmitter's activities. Consequently, dealing with super-reactive jammers is very challenging. Thus, we introduce a smart deception mechanism to attract the jammer to continuously attack the channel and then leverage jamming signals to transmit data based on the ambient backscatter communication technology. To detect the backscattered signals, the maximum likelihood detector can be adopted. However, this method is notorious for its high computational complexity and requires the model of the current propagation environment as well as channel state information. Hence, we propose a deep learning-based detector that can dynamically adapt to any channels and noise distributions. With a Long Short-Term Memory network, our detector can learn the received signals' dependencies to achieve a performance close to that of the optimal maximum likelihood detector. Through simulation and theoretical results, we demonstrate that with our approaches, the more power the jammer uses to attack the channel, the better bit error rate performance the transmitter can achieve.

\end{abstract}


%====================================================================================
%====================================================================================
\section{Introduction}
\label{Sec:intro}

Wireless communications play an essential role in many areas by facilitating tetherless and ubiquitous transmissions through the broadcast medium. However, due to the exposed nature of wireless communications, current wireless networks are extremely vulnerable to jamming attacks. In particular, the jammer intentionally injects high-power interference signals to the target wireless channels to significantly reduce the signal-to-interference-plus-noise ratio (SINR) at the legitimate receiver, and thus it can disrupt or even bring down the whole system. Among all types of jamming attacks, the super-reactive jamming attack is one of the most difficult jamming attacks to deal with in practice. In particular, the super-reactive jammer has unlimited power budget and is equipped with the self-interference suppression (SiS) capability~\cite{7524449} ~\cite{Afifi2013Exploiting},~\cite{Hanawal2020Cognitive}allowing it to simultaneously attack and listen to the transmitter's activities to adjust its attack strategy. Consequently, all current anti-jamming approaches such as frequency hopping~\cite{Xiao2012Jamming, Gao2018Game}, rate adaption~\cite{Noubir2011RA}, and deception~\cite{Hoang2015Performance},~\cite{Hoang2020Borrowing} cannot effectively defeat such an attack since with virtually unlimited power, the super-reactive jammer can theoretically jam all channels over which it discerns activity of legitimate users (thanks to the jammer's SiS capability).

In this paper, we introduce a novel anti-jamming framework to effectively defeat super-reactive jamming attacks. This framework is composed of three main components: (i) intelligent deception strategy, (ii) smart multi-source backscatter communications, and (iii) a smart detection mechanism. Specifically, when a legitimate transmitter detects a reactive jamming attack on its communications channel, it will pretend to be oblivious and continue the data transmission to encourage the jammer to continue its attack. Then, the transmitter changes the communication method by activating the backscatter communications mode to transmit data instead of using the current active transmission mode. The ambient backscatter communication technology, unlike conventional wireless communications, can work well under the presence of radio jamming interference as it can reflect or absorb the jamming signals to backscatter information to the receiver~\cite{Huynh2018Survey,Huynh2019Jam}. To improve the backscattering performance, the receiver can be equipped with multiple antennas~\cite{Parks2014Turbo}. To be more specific, the transmitter will use an auxiliary low-power tag device to backscatter its information simultaneously on its own transmitted signals \emph{and} the jamming signals, referred to as multi-source backscattering communications. As such, this solution allows the system to communicate under super-reactive jamming attacks. Finally, to enhance the detection efficiency at the receiver, we develop a novel deep learning (DL) based detection mechanism. It is widely known that the maximum likelihood (ML) detection method is optimal in terms of bit error rate (BER). However, the ML detector requires the model of the current propagation environment as well as channel state information~\cite{Fan2019CNN}. As such, it cannot be directly applied to other scenarios with different channel and noise distributions (e.g., when the jammer adjusts its attack strategy or other anti-jamming systems in different wireless environments). Moreover, it is very difficult to obtain the channel coefficients of the links from the jammer as the jammer will not send a known training sequence. Instead, our proposed DL-based detector does not require this information, and thus it is robust the wireless environment and jammer's attack strategies. As a result, the proposed DL detector can be adopted in various anti-jamming systems with different channel distributions as well as different jamming attack strategies.

Through simulation results and theoretical analysis, we demonstrate that, under specific assumptions, the proposed framework is able to not only successfully defeat the super-reactive jammer, but also to leverage the jamming signals to enhance the BER performance of the system. In particular, we first show that with our proposed framework, the backscatter rate and BER performance increases with the jamming power of the super-reactive jammer. As such, the proposed framework can defeat the super-reactive jamming attacks. Furthermore, the proposed DL-based detector can achieve the BER performance close to that of the optimal ML detector without requiring information of the jammer in advance (e.g., jamming signals' characteristics and channel coefficients). Finally, through simulation results, it it revealed that the system performance can be significantly improved if we use multiple antennas at the legitimate receiver.

The rest of this paper is organized as follows. In Section~\ref{sec:relatedwork}, we give an overview of related work in anti-jamming and signal detection. Section~\ref{Sec.System} and Section~\ref{Sec.signalmodel} present the system model and the channel model, respectively. Section~\ref{sec:decoding} discusses the decoding process at the receiver. Section~\ref{sec:deeplearning} introduces our proposed DL-based signal detector. After that, the evaluation results are discussed in Section~\ref{sec:evaluation}. Finally, conclusions are drawn in Section~\ref{sec:conclusion}.

%=================================================
%=================================================
\section{Related Work and Main Contributions}
\label{sec:relatedwork}

\subsection{Anti-Jamming Approaches}

Several anti-jamming approaches have been proposed in the literature~\cite{Mpitziopoulos2009Survey}. Although these approaches are feasible in coping with conventional jamming attacks, e.g., proactive jamming, they are not effective or even impractical in defeating super-reactive jamming attacks. In this subsection, we first highlight current anti-jamming techniques and their drawbacks in coping with the super-reactive jamming.

\subsubsection{Rate Adaptation}

One common anti-jamming solution is the rate adaptation (RA) technique~\cite{Noubir2011RA}. This technique allows the transmitter to adapt its transmission rate under jamming attacks. For example, when the jammer attacks the channel at a high-power level, the transmitter can reduce its transmission rate to hide itself from the jammer. In contrast, the transmitter can increase its transmission rate when the jamming power is low. However, in~\cite{Noubir2011RA}, the authors demonstrated that the RA technology is not effective in dealing with smart jamming like reactive jamming, especially when the jamming power is very high. Additionally, the authors pointed out that in some cases, the transmitter fails to recover from the lowest data rate even when the jammer ceases its attacks. For the super-reactive jammer, the RA technique is not effective as the jammer can always attack the system with high jamming power. More importantly, by simultaneously listening to activities of the transmitter while jamming by using recent advances in signal detection and sensing, e.g., self-interference suppression technology~\cite{Afifi2013Exploiting},~\cite{Hanawal2020Cognitive}, the jammer can quickly adapt its attack power as soon as the transmitter changes its transmission rate. Consequently, the RA technology is not suitable for dealing with super-reactive jammers.

\subsubsection{Frequency Hopping}

Another anti-jamming solution widely adopted in existing studies is the frequency hopping (FH) technology~\cite{Xiao2012Jamming, Gao2018Game}. The key idea of this solution is that the transmitter can switch to another channel when the jammer attacks the current channel. However, these solutions and other FH approaches in the literature inherit various drawbacks in defeating jamming attacks, especially under reactive jamming. First, the FH technology requires a pre-shared key at both the transmitter and the receiver, resulting in the problem of scalability due to the need of distributing the shared key~\cite{Xiao2012Jamming}. Second, under the reactive jamming, when the transmitter switches to another communication channel, the reactive jammer can quickly detect and attack the new channel. Finally, if the jammer's power budget is very high, e.g., the super-reactive jammer considered in this paper, it can attack all the channels at the same time, and thus the FH technology is no longer effective.

\subsubsection{Recent Proposals}

Recently, several solutions have been proposed to overcome the limitations of conventional anti-jamming mechanisms. Specifically, in~\cite{Gu02017Exploiting}, when the jammer attacks the channel, the transmitter (equipped with the radio frequency (RF) energy harvesting capability) can harvest RF energy from the jamming signals. Then, the harvested energy is used to support the internal operations and active transmissions when the jammer is idle. Differently, the authors in~\cite{Huynh2019Jam} proposed to use the ambient backscatter communications~\cite{Huynh2018Survey},~\cite{Liu2013Ambient} to allow the transmitter to transmit information under jamming attacks by backscattering the jamming signals. The authors demonstrated that by leveraging the jamming signals, the average throughput of the transmitter increases with the jamming power. However, all these solutions are only suitable for dealing with proactive jammers, i.e., the jammers attack the channel regardless of the transmitter's activities. In contrast, the reactive jammer can quickly cease its attacks when there is no active transmission from the transmitter. As such, the performance of these solutions in dealing with the super-reactive jammer is very poor.

Several solutions have been proposed recently to deal with reactive jammers. In~\cite{Hoang2015Performance}, the authors proposed a deception tactic to deal with jamming attacks in a cognitive radio communication system. The key idea of this solution is that the secondary user can generate fake signals to lure/trap the jammer to attack the channel, and thus it can undermine the attack abilities by draining the power of the jammer. As the authors considered scenarios in which the jammer has a limited power budget, this solution works well because the jammer does not have enough power to attack when the secondary user transmits the actual information. However, if the jammer's power budget is high (e.g., the super-reactive jammer considered in this paper), this solution is no longer efficient. In~\cite{Hoang2020Borrowing}, the authors proposed a deception-based anti-jamming framework for IoT networks. In particular, the transmitter can decide either to perform deception or actively transmit data at the beginning of the time slot. If the transmitter chooses to perform deception, it will first generate ``fake'' signals to lure the jammer. After that, as soon as the jammer attacks the channel, the transmitter can harvest RF energy from the jamming signals or backscatter information to the receiver through the jamming signals. Nevertheless, for reactive jammers equipped with the self-interference suppression technology~\cite{Hanawal2020Cognitive}, this solution is not effective. The reason is that the reactive jammer can simultaneously attack the channel and listen to the activities of the transmitter. As soon as the transmitter ceases its active transmissions, the jammer will stop its jamming attacks, and thus there is no jamming signals for the transmitter to leverage, i.e., harvest energy and backscatter information. To the best of our knowledge, all existing anti-jamming solutions cannot effectively deal with reactive jammers, especially the super-reactive jammer which has self-interference suppression ability and is equipped with an unlimited power budget.

\subsection{Deep Learning for Signal Detection}

There are some research works using deep learning for signal detection~\cite{Fan2019CNN, Ye2018Power, Shea2018Over, Rajendran2018Deep}. In~\cite{Ye2018Power}, the authors proposed a deep neural network to estimate the wireless channel and detect signals in orthogonal frequency-division multiplexing systems. In~\cite{Shea2018Over}, the authors focused on signal detection for radio communication signals by using a one-dimension convolutional neural network architecture. Nevertheless, to the best of our knowledge, there is no work considering deep learning for signal detection in ambient backscatter communication systems under jamming attacks. Detecting backscattered bits in our work is much more challenging compared with conventional active transmissions due to the following reasons. Firstly, the backscattered signals are very weak with low data rates compared to traditional active transmissions. As such, it is difficult for the receiver to detect such weak signals with high accuracy (usually higher than 99\% for reliable communications). Secondly, the received signals at the receiver are a combination of the backscattered signals from the backscatter tag and the direct signals from the jammer as well as the transmitter. Consequently, it is very challenging for the receiver to distinguish the backscattered signals from jamming signals and the transmitter's signals. To address this problem, in this work we propose to adopt the LSTM network, which is designed to learn the sequence and time-series data, to learn the long-term dependencies among the received signals at different antennas over several RF source symbols for each backscattered bit. By doing this, the receiver can better estimate the weak backscattered signals because the LSTM network can store information learned from previous input data (i.e., received signals) and combine them to achieve better estimation performance. Moreover, the proposed DL-based signal detector can be straightforwardly applied to other scenarios, e.g., when the reactive jammer changes its attack strategy or other anti-jamming systems in different wireless environments. To the best of our knowledge, the proposed detector is the first solution using deep learning in signal detection for ambient backscatter communications in which the backscattered signals are usually weak.

\subsection{The Novelty and Main Contributions}

The major contributions of this paper can be summarized as follows.

\begin{itemize}
	
	\item We propose a novel anti-jamming framework to defeat a super-reactive jammer that has unlimited power budget and is equipped with the self-interference suppression capability to simultaneously attack and listen to the transmitter's activities. Specifically, when the transmitter detects jamming attacks on the channel, it will pretend to be oblivious and continue to transmit data to lure/trap the jammer. Then, the transmitter switches to the backscatter communication mode by using the ambient backscatter tag to send the actual information by simultaneously backscattering the jammer's signals and the transmitter's signals. As such, the transmitter can communicate with the receiver even under very strong jamming attacks. To the best of our knowledge, this is the first anti-jamming framework that can efficiently deal with super-reactive jamming attacks.
	
	\item We present the signal models of direct links (from the jammer and the transmitter) and the backscattered link (from the tag). Given the fact that a closed-form expression of the capacity of the backscatter channel remains unknown, we also develop a numerical method to quantify the maximum achievable backscatter rate of the ambient backscatter tag. 
	
	\item We propose a novel DL-based signal detector by using the LSTM network at the backscatter receiver. The proposed DL-based signal detector does not require the channel state information and the model of the current propagation environment to detect the backscattered bits as in the conventional maximum likelihood detector. As such, our proposed DL-based signal detector can be straightforwardly applied to other scenarios, e.g., when the reactive jammer changes its attack strategy or other anti-jamming systems in different wireless environments. To the best of our knowledge, the proposed detector is the first solution using deep learning in signal detection for ambient backscatter communications.
	
	\item Finally, we perform extensive simulations to demonstrate the efficiency of proposed solutions. In particular, the proposed framework can allow the transmitter to transmit data even under very strong and super-reactive jamming attacks. Importantly, by using our proposed solution, the BER performance of the system can be enhanced as the attack power of the jammer increases.
	
\end{itemize}

\underline{\textit{Notation:}} The notations used in this paper are listed in the following. Lowercase letter $g$ denotes a scalar variable. Boldface lowercase letter $\mathbf{g}$ denotes a vector. Boldface uppercase letter $\mathbf{G}$ denotes a matrix. $\mathcal{CN}(\mu, \xi^2)$ is the circularly symmetric complex Gaussian (CSCG) distribution with variance $\xi^2$ and mean $\mu$. $\mathbf{I}_M$ is the $M \times M$ identical matrix. $\mathbf{G}^\mathsf{T}$ denotes the transpose of matrix $\mathbf{G}$. $\mathbf{G}^\text{H}$ is the conjugate transpose of matrix $\mathbf{G}$. $\mathbb{E}[.]$ is the statistical expectation. $\odot$ denotes the Hadamard product, i.e., element-wise multiplication of vectors. $I(X;Y)$ denotes the mutual information of random variables $X$ and $Y$. $H(X|Y)$ is the conditional entropy. $\oplus$ is the addition modulo-2 operator. All subscripts with italic fonts represent indicators such as iteration and time variables. Superscripts $\mathrm{t}$, $\mathrm{j}$, $\mathrm{r}$, and $\mathrm{b}$ with straight fonts stand for the transmitter, the jammer, the receiver, and the backscatter tag, respectively.

%====================================================================================
%====================================================================================
\section{Overview of The Proposed Anti-jamming System}
\label{Sec.System}
\begin{figure}[!]
	\centering
	\includegraphics[scale=0.27]{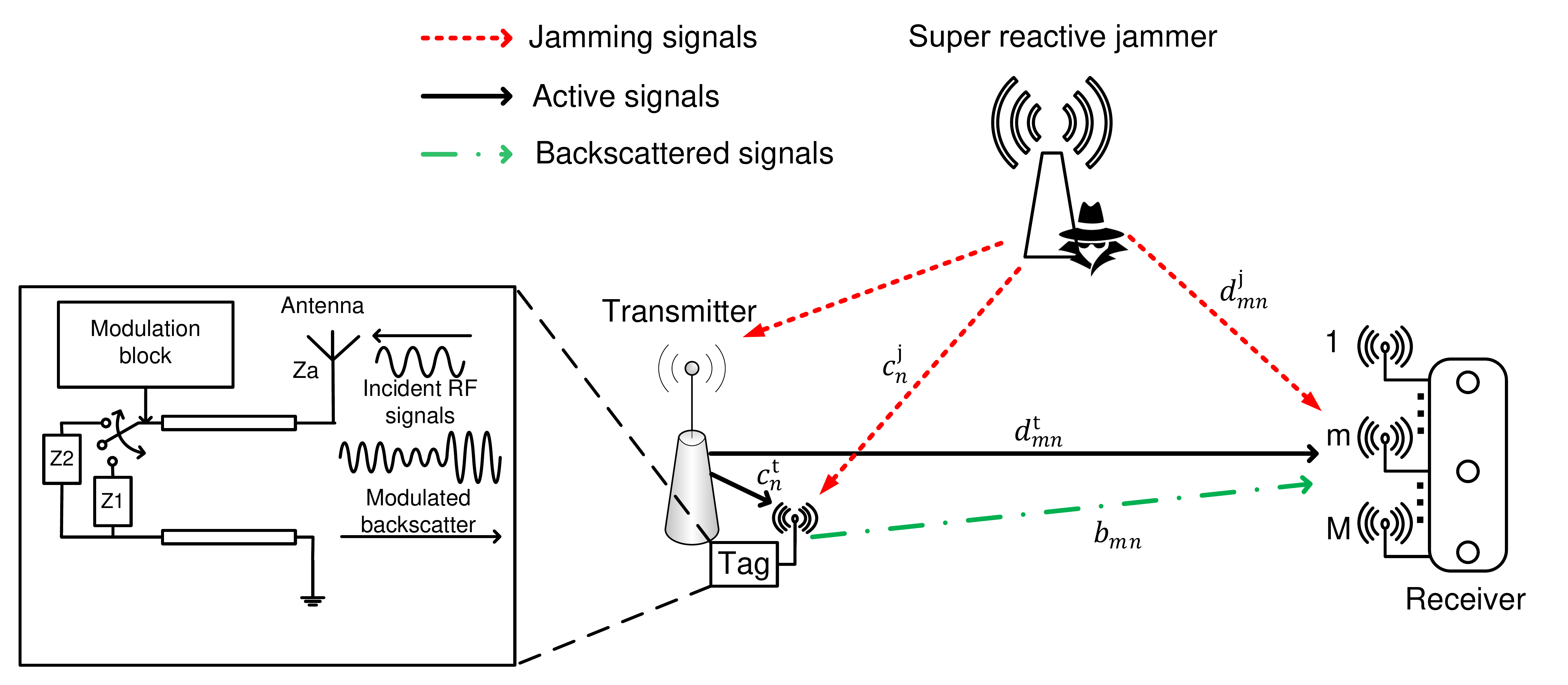}
	\caption{System model.}
	\label{Fig.system_model}
\end{figure}

In this paper, a legitimate wireless communication system including a transmitter and a receiver is considered, as illustrated in Fig.~\ref{Fig.system_model}. In this paper, we assume that the locations of transmitter, the receiver, and the jammer are fixed. The transmitter transmits data to the receiver on a dedicated channel. We assume that the receiver has more energy and computational capacity than the transmitter. A super-reactive jammer, which is located near the legitimate system, aims to disrupt/bring down the communications between the transmitter and the receiver by injecting strong interference signal to the dedicated channel. Note that this work aims to defeat a super-reactive jammer with unlimited power budget. Thus, we assume that once the jammer attacks the channel, it will attack with the highest power level. Consequently, all transmitted packets by active transmissions will be lost.

\subsection{Jammer Model}

We consider a super-reactive jammer that is equipped with an unlimited power budget and can simultaneously attack the channel and listen to the activities of the transmitter thanks to the development of the self-interference suppression technique~\cite{Hanawal2020Cognitive,Afifi2013Exploiting}. As such, the super-reactive jammer is able to detect the status of the transmitter to cease the jamming attacks when the transmitter stops transmitting on the channel or quickly launch jamming attacks when the transmitter actively transmits the very first bits on the channel. Given such type of jamming attacks, all current anti-jamming approaches might not be effective to implement. For example, deception mechanisms proposed in~\cite{Hoang2015Performance},~\cite{Hoang2020Borrowing},~\cite{Huynh2020Deep} are no longer effective because the reactive jammer can quickly stop its attacks when the transmitter finishes the deception period, and thus there are no jamming signals for the transmitter to leverage. Moreover, with the unlimited power budget, the super-reactive jammer can simultaneously inject high-power interference to multiple channels. Hence, conventional anti-jamming approaches like rate adaptation and frequency hopping technologies~\cite{Noubir2011RA, Xiao2012Jamming} also cannot effectively defeat super-reactive jamming attacks, as discussed in the previous section. To the best of our knowledge, the super-reactive jammer considered in this work is the most difficult jamming attacks to deal with in practice, and our proposed framework is the first work in the literature which can deal with such kind of attacks. In this paper, we assume that the reactive jammer is equipped with a single antenna. In case the reactive jammer would be equipped with multiple antennas, it would be more challenging. This can be considered as a potential future research direction from our current work.

\subsection{Transmitter Design}

To defeat the reactive jammer, at the transmitter's side, we implement an auxiliary low-power ambient backscatter tag equipped with ambient backscatter communication technology, as illustrated in Fig.~\ref{Fig.system_model}. In particular, if the dedicated channel is under jamming attacks, the transmitter keeps its transmissions to attract the reactive jammer\footnote{Note that the transmitter considered in this paper is equipped with a single antenna. However, our proposed approach can be directly applied to other scenarios in which the transmitter is equipped with multiple antennas.}. Concurrently, the transmitter switches the transmission mode by using the ambient backscatter tag. Then, the tag uses the ambient backscatter communication technology to send the actual data by backscattering the transmitter's signals and the jammer's signals. In other words, the jammer and the transmitter are considered as two ambient RF sources that are leveraged to backscatter information to the receiver~\cite{Yang2017Riding}. In practice, the jamming signal is usually much stronger than unintentional interference from surrounding wireless devices, especially for the super-reactive jammer considered in this work. The reason is that the jammer aims to send the strong jamming signals to the channel to reduce the signal-to-noise ratio at the receiver. Thus, we can easily detect the jamming attack by measuring the signals on the channel and/or based on the feedback information from the receiver. Another method to detect jamming attacks is based on the packet delivery ratio (PDR) which is the ratio between the number of packets that has been successfully delivered and acknowledged by the receiving node to the number of packets sent by the transmitting node~\cite{Hanawal2016Jamming}. Under a jamming attack, the PDR will be significantly degraded, and thus the transmitter can detect the jamming attacks. Note that even in the case of miss-detection, our proposed solution still can work well as the tag can backscatter information through the transmitter's signals.

To be more specific, the ambient backscatter tag modulates the information sent from the transmitter and then backscatters the transmitter's signals and the jammer's signals to send the information to the receiver. To that end, the ambient backscatter tag is equipped with an RF switch, e.g., ADG902~\cite{Huynh2018Survey}, that is directly connected to the antenna of the tag, as illustrated in Fig.~\ref{Fig.system_model}. The bit stream sent from the transmitter is first encoded by using a differential encoding technique, i.e., modulo-2. When the tag wants to transmit bits ``0'', it will switch to load $Z_1$ to change to the non-reflecting state (also known as the absorbing state). In contrast, the tag switches to load $Z_2$ to change to the reflecting state to transmit bits ``1''. In this way, the information can be backscattered to the receiver by just switching between the non-reflecting and reflecting states. Thus, unlike conventional wireless communications, the ambient backscatter communication technology can work well under radio interference/jamming as it can absorb or reflect the jamming signals to send data to the receiver. This is a fundamental principle of backscatter communication systems, and interested readers can find more information from the following references~\cite{Huynh2018Survey},~\cite{Liu2013Ambient}. It is worth noting that the ambient backscatter tag conveys information to the receiver by reflecting/observing the ambient RF signals instead of generating active RF signals like the conventional active transmission. As a result, the ambient backscatter communication technology can work well under the presence of multiple RF signals in the same environment~\cite{Yang2017Riding}. In this way, in our proposed framework, the tag device can leverage the RF signals from the transmitter and jammer to transmit data to the receiver. The process of backscattering and detecting information under super-reactive jamming attacks will be presented in details in Section~\ref{Sec.signalmodel} and Section~\ref{sec:decoding}.

It is worth noting that backscatter communications may not always outperform active transmissions with on-off keying (OOK). In particular, when the jammer attacks the channel with low power levels and the transmitter transmits information with high-power signals, the receiver still can detect the transmitted information as it can easily distinguish between the jamming signals and the active signals from the transmitter. Instead, if the tag backscatters information through the jamming signals and the active signals, the backscatter rate will be lower than the active transmission rate for the following reasons. Firstly, to be able to detect the backscattered signals at the receiver, the tag must backscatter information at a lower rate compared to that of the active transmissions. Secondly, with weak jamming signals, the tag may not be able to backscatter information at reasonable backscatter rates. In contrast, if the jammer attacks the channel with high power levels, the proposed backscatter solution is more beneficial than active transmissions. Specifically, with high jamming power, the receiver may not be able to detect the active signals, even if the transmitter transmits high-power signals. Note that, the active transmissions may be successful if the transmitter transmits information with a much higher power level compared to that of the jammer. However, this is usually not feasible in practice. Moreover, with strong jamming signals, the tag can backscatter information at a high rate, resulting in a good performance compared to that of normal active transmissions.

\subsection{Receiver Design}

As mentioned, the received signals contain the backscattered signals and the direct link signals from the transmitter and the jammer, resulting in a low signal detection accuracy. To address this problem, we employ $M$ antennas\footnote{Our proposed approach still works well even the receiver equipped with a single antenna. Results for this case can be found in the simulation section.} at the receiver. As such, the receiver can exploit diversity gain to improve the BER performance. More importantly, we propose a deep learning-based signal detector that can apply to any channel and noise distribution and achieve the BER performance close to that of the optimal ML detector. The process of decoding the backscattered signals will be presented in Section~\ref{sec:decoding} and Section~\ref{sec:deeplearning} in details. It is worth noting that the receiver can also detect the jamming attacks from the jammer and switch its reception mode to receive and detect the backscattered signals. Moreover, the proposed model and analysis in this paper can be straightforwardly adopted in other scenarios with multiple backscatter tags. To that end, the receiver can deploy popular scheduling mechanisms to eliminate the collisions between tags. Another potential solution is backscattering data at different backscatter rates~\cite{Huynh2018Survey},~\cite{Liu2013Ambient}. In this way, the receiver can easily distinguish the backscattered signals from different tags. Note that when the jamming signals and the transmitter's signals are strong, it is more difficult for the receiver to detect signals. To solve this problem, a two-stage hybrid analog-digital design proposed in~\cite{Duan2019Hybrid, Duan2020Ambient} can be adopted to eliminate the effect of the large dynamic range on detecting signals. More importantly, our proposed approach can also directly apply to scenarios with multiple jammers. In particular, these jammers are also considered as ambient RF sources to support the backscatter communications of the tag with the same signal model and analysis of the jammer considered in this paper. One example of backscattering from multiple sources and signal patterns can be found in~\cite{Yang2017Riding}.

%====================================================================================
%====================================================================================
\section{Channel Model}
\label{Sec.signalmodel}
In this section, we provide details on the channel model and analyze the received signals that will be used to develop the proposed detectors. In particular, to make favorable conditions for the receiver to successfully detect the backscattered signals from the tag, the backscatter rate should be lower than the symbol rate of ambient RF source signals~\cite{Huynh2018Survey},~\cite{Liu2013Ambient}. As such, in this work, we assume that the ambient backscatter tag backscatters data at the rate $N$ times lower than the symbol rate of the ambient RF signals~\cite{Guo2019Nocoherent,Zhang2019Constellation}. In other words, each backscatter symbol will be backscattered over $N$ RF source symbols. Note that, we do not use any channel coding technique in this work for simplicity. In the system under consideration, the received signals at the receiver include the direct link signals from the transmitter and the jammer, the backscattered signals from the tag, and noise. Denote $y_{mn}$ as the $n$-th received sample at the $m$-th antenna of the receiver. We have
\begin{equation}
\label{eq:totalreceive}
y_{mn} = \underbrace{d_{mn}^{\mathrm{t}} + d_{mn}^{\mathrm{j}}}_{\text{direct links}} + \underbrace{b_{mn}}_{\text{backscatter links}} + \sigma_{mn},
\end{equation}
where $d_{mn}^{\mathrm{t}}$ and $d_{mn}^{\mathrm{j}}$ are the direct link signal from the transmitter and the jammer, respectively. $b_{mn}$ is the backscattered signal from the backscatter tag and $\sigma_{mn}$ is the CSCG noise with zero mean and unit variance, i.e., $\sigma_{mn} \sim \mathcal{CN}(0,1)$.

\subsection{Direct Links}

Denote $s_n^\mathrm{j}$ and $s_n^\mathrm{t}$ as the signals generated from the jammer and the transmitter at time instant $n$, respectively. Similar to other research works in the literature~\cite{Guo2019Nocoherent},~\cite{Zhang2019Constellation}, $s_n^\mathrm{j}$ and $s_n^\mathrm{t}$ are assumed to be independent and identically distributed (i.i.d) at every time instant $n$. Note that in this paper, the tag considers the transmitter as an ambient RF source to support its backscatter communications. We assume that the RF signals generated from the transmitter during the jamming attacks are unknown at the receiver and are random (e.g., the transmitter actively transmits random/blank signals to attract the jammer and save energy). Thus, we assume that $s_n^\mathrm{t}$ follows the standard CSCG distribution, i.e., $s_n^\mathrm{t} \sim \mathcal{CN}(0,1)$~\cite{Qian2017Nocoherent}. Moreover, similar to other studies~\cite{Xing2015To}-\cite{Pirzadeh2016Subverting}, we assume that the jamming signals $s_n^\mathrm{j}$ follow the standard CSCG distribution as the jamming signals are also usually random, i.e., $s_n^\mathrm{j} \sim \mathcal{CN}(0,1)$.\footnote{Note that our proposed deep learning-based detector does not require this information in advance. The proposed deep learning-based detector can be directly applied for any other channel distributions.}

Then, the direct link signals from the jammer received at the $m$-th antenna of the receiver are expressed as follows~\cite{Guo2019Nocoherent,Zhang2019Constellation}:
\begin{equation}
\label{eq:directjammer}
d_{mn}^{\mathrm{j}} = f_m^\mathrm{j} \sqrt{P^\mathrm{jr}}s_n^\mathrm{j},
\end{equation}
where $f_m^\mathrm{j}$ denotes the Rayleigh fading~\cite{Zhang2019Constellation} from the jammer to the receiver. Without loss of generality, let $\mathbb{E}[|f_m^\mathrm{j}|^2] = 1$. $P^\mathrm{jr}$ presents the average received power of the direct link from the jammer to the receiver. $P^\mathrm{jr}$ can be obtained as follows~\cite{Rappaport2002Wireless}:
\begin{equation}
P^\mathrm{jr} = \frac{ \kappa P^\mathrm{j} G^\mathrm{j} G^\mathrm{r} }{{(L^\mathrm{jr})}^{\upsilon^\mathrm{jr}}},
\end{equation}
where $P^\mathrm{j}$ is the transmit power of the jammer, $G^\mathrm{j}$ is the antenna gain of the jammer, $G^\mathrm{r}$ is the antenna gain of the receiver, and $L^\mathrm{jr}$ is the distance between the receiver and the jammer. $\upsilon^\mathrm{jr}$ is the path loss exponent and $\kappa = {(\frac{\lambda}{4\pi})}^2$ with $\lambda$ is the wavelength.

In the same way, the direct link signals from the transmitter received at the $m$-th antenna of the receiver are also expressed as follows:
\begin{equation}
\label{eq:directtransmitter}
d_{mn}^{\mathrm{t}} = f_m^\mathrm{t} \sqrt{P^\mathrm{tr}}s_n^\mathrm{t},
\end{equation}
where $f_m^\mathrm{t}$ denotes the Rayleigh fading~\cite{Zhang2019Constellation} from the transmitter to the receiver. Without loss of generality, let $\mathbb{E}[|f_m^\mathrm{t} |^2] = 1$. $P^\mathrm{tr}$ is the average received power of the direct link from the transmitter to the receiver. $P^\mathrm{tr}$ can be calculated as follows:
\begin{equation}
P^\mathrm{tr} = \frac{ \kappa P^\mathrm{t} G^\mathrm{t} G^\mathrm{r}}{ {(L^\mathrm{tr})}^{\upsilon^\mathrm{tr}}},
\end{equation}
where $P^\mathrm{t}$ presents the transmit power of the transmitter, $L^\mathrm{tr}$ denotes the distance between the receiver and the transmitter, $G^\mathrm{t}$ is the antenna gain of the transmitter, and $\upsilon^\mathrm{tr}$ is the path loss exponent.

\subsection{Backscatter Link}

The ambient backscatter tag can backscatter the RF signals from both the reactive jammer and the transmitter to send data to the receiver by using the ambient backscatter communication technology~\cite{Huynh2018Survey,Liu2013Ambient}. In the following, we formulate the backscattered signals received at the receiver in details. Specifically, the jamming signals received at the backscatter tag are expressed as follows:
\begin{equation}
\label{eq:backjammer}
c_n^\mathrm{j} = g^\mathrm{j} \sqrt{P^\mathrm{jb}}s_n^\mathrm{j},
\end{equation}
where $g^\mathrm{j}$ denotes the Rayleigh fading~\cite{Zhang2019Constellation} from the jammer to the backscatter tag. Without loss of generality, let $\mathbb{E}[|g^\mathrm{j}|^2] = 1$. $P^\mathrm{jb}$ is the average jamming power received at the backscatter tag. $P^\mathrm{jb}$ is calculated as follows:
\begin{equation}
P^\mathrm{jb} = \frac{\kappa P^\mathrm{j} G^\mathrm{j} G^\mathrm{b}}{{(L^\mathrm{jb})}^{\upsilon^\mathrm{jb}}},
\end{equation}
where $L^\mathrm{jb}$ denotes the distance between the backscatter tag and the jammer, $\upsilon^\mathrm{jb}$ is the path loss exponent, and $G^\mathrm{b}$ is the tag's antenna gain. Similarly, the transmitter's signals received at the backscatter tag are derived as follows:
\begin{equation}
\label{eq:backtransmitter}
c_n^\mathrm{t} = g^\mathrm{t} \sqrt{P^\mathrm{tb}}s_n^\mathrm{t},
\end{equation}
where $g^\mathrm{t}$ is the Rayleigh fading~\cite{Zhang2019Constellation} from the transmitter to the tag. Without loss of generality, let $\mathbb{E}[|g^\mathrm{t}|^2] = 1$. $P^\mathrm{tb}$ denotes the average power from the transmitter received at the backscatter tag and can be calculated as follows:
\begin{equation}
P^\mathrm{tb} = \frac{\kappa P^\mathrm{t} G^\mathrm{t} G^\mathrm{b}}{{(L^\mathrm{tb})}^{\upsilon^\mathrm{tb}}},
\end{equation}
where $L^\mathrm{tb}$ is the distance between the tag and the transmitter and $\upsilon^\mathrm{tb}$ is the path loss exponent.

Denote $c_n$ as the total RF signals received at the backscatter tag. From (\ref{eq:backjammer}) and (\ref{eq:backtransmitter}), we have
\begin{equation}
\label{eq:cn}
c_n = c_n^\mathrm{t} + c_n^\mathrm{j} = g^\mathrm{t} \sqrt{P^\mathrm{tb}}s_n^\mathrm{t} + g^\mathrm{j} \sqrt{P^\mathrm{jb}}s_n^\mathrm{j}.
\end{equation}

From the fundamental of the ambient backscatter communication technology~\cite{Liu2013Ambient},~\cite{Huynh2018Survey}, the backscatter tag has two states (i.e., \textit{OOK}) including the reflecting state denoted by $e=1$ and the non-reflecting state (i.e., absorbing state) denoted by $e=0$. As mentioned, the backscatter rate must be $N$ times lower than the symbol rate of the RF sources. Thus, state $e$ remains unchanged during $N$ source symbols. Denote $s_n^\mathrm{b}$ as the backscattered signals from the backscatter tag, we have
\begin{equation}
\label{eq:sbn}
s_n^\mathrm{b} = \gamma c_n e, 
\end{equation}
where $\gamma$ denotes the reflection coefficient which combines all the effects of the load impedance, antenna impedance, and other components at the ambient backscatter tag~\cite{Zhang2019Constellation}.

The backscatter then backscatters signals to the receiver to send the actual information. We then formulate the backscattered signals received at the $m$-th antenna of the receiver as follows:
\begin{equation}
\label{eq:bn}
b_{mn} = f_m^\mathrm{b} \sqrt{\frac{G^\mathrm{b} G^\mathrm{r} \kappa}{{(L^\mathrm{br})}^{\delta}}} s_n^\mathrm{b},
\end{equation}
where $L^\mathrm{br}$ denotes the distance between the receiver and the backscatter tag, $\delta$ is the path loss exponent, and $f_m^\mathrm{b}$ is the Rayleigh fading~\cite{Zhang2019Constellation} from the tag to the receiver. Without loss of generality, let $\mathbb{E}[|f_m^\mathrm{b}|^2] = 1$. Substituting (\ref{eq:cn}) and (\ref{eq:sbn}) into (\ref{eq:bn}), we have
\begin{equation}
\label{eq:backscattersignal}
\begin{aligned}
b_{mn} &= f_m^\mathrm{b} \sqrt{\frac{G^\mathrm{b} G^\mathrm{r} \kappa}{{(L^\mathrm{br})}^{\delta}}} \gamma e \Big(g^\mathrm{t} \sqrt{P^\mathrm{tb}}s_n^\mathrm{t} + g^\mathrm{j} \sqrt{P^\mathrm{jb}}s_n^\mathrm{j}\Big)\\
& = f_m^\mathrm{b} e \Big(g^\mathrm{t} \sqrt{\frac{\kappa |\gamma|^2 P^\mathrm{tr}{(G^\mathrm{b})}^2{(L^\mathrm{tr})}^{\upsilon^\mathrm{tr}}}{{(L^\mathrm{tb})}^{\upsilon^\mathrm{tb}}{(L^\mathrm{br})}^{\delta}}} s_n^\mathrm{t} + g^\mathrm{j} \sqrt{\frac{\kappa |\gamma|^2 P^\mathrm{jr}{(G^\mathrm{b})}^2{(L^\mathrm{jr})}^{\upsilon^\mathrm{jr}}}{{(L^\mathrm{jb})}^{\upsilon^\mathrm{jb}}{(L^\mathrm{br})}^{\delta}}} s_n^\mathrm{j}\Big).
\end{aligned}
\end{equation}
Denote $\tilde{\alpha}^\mathrm{t} = \frac{\kappa {|\gamma|^2}{(G^\mathrm{b})}^2{(L^\mathrm{tr})}^{\upsilon^\mathrm{tr}}} {{(L^\mathrm{tb})}^{\upsilon^\mathrm{tb}} {(L^\mathrm{br})}^{\delta}}$ and $\tilde{\alpha}^\mathrm{j} = \frac{\kappa {|\gamma|^2}{(G^\mathrm{b})}^2{(L^\mathrm{jr})}^{\upsilon^\mathrm{jr}}} {{(L^\mathrm{jb})}^{\upsilon^\mathrm{jb}} {(L^\mathrm{br})}^{\delta}}$, we can rewrite (\ref{eq:backscattersignal}) as follows:
\begin{equation}
\label{eq:totalbackscatter}
b_{mn} = f_m^\mathrm{b} e \Big(g^\mathrm{t} \sqrt{\tilde{\alpha}^\mathrm{t} P^\mathrm{tr}}s_n^\mathrm{t} + g^\mathrm{j} \sqrt{\tilde{\alpha}^\mathrm{j} P^\mathrm{jr}}s_n^\mathrm{j} \Big).
\end{equation}

\subsection{Received Signals at the Receiver}
\label{sec:received}
Substituting (\ref{eq:directtransmitter}), (\ref{eq:directjammer}), and (\ref{eq:totalbackscatter}) to (\ref{eq:totalreceive}), we have
\begin{equation}
\label{eq:receSignal}
\begin{aligned}
y_{mn} &= d_{mn}^{\mathrm{t}} + d_{mn}^{\mathrm{j}} + b_{mn} + \sigma_{mn}\\
&= f_m^\mathrm{t} \sqrt{P^\mathrm{tr}}s_n^\mathrm{t} + f_m^\mathrm{j} \sqrt{P^\mathrm{jr}}s_n^\mathrm{j} + f_m^\mathrm{b} e \Big(g^\mathrm{t} \sqrt{\tilde{\alpha}^\mathrm{t} P^\mathrm{tr}}s_n^\mathrm{t} + g^\mathrm{j} \sqrt{\tilde{\alpha}^\mathrm{j} P^\mathrm{jr}}s_n^\mathrm{j} \Big) + \sigma_{mn}.
\end{aligned}
\end{equation}
In this paper, we assume that the sampling rate is low, e.g., 10 MHz, and the distance between the transmitter and the receiver is not too long, e.g., less than 30 meters. As such, the time delay between the receptions of $d_{mn}^{\mathrm{t}}$, $d_{mn}^{\mathrm{j}}$, and $b_{mn}$ at the receiver is negligible and can be ignored~\cite{Liu2013Ambient},~\cite{Qian2017Nocoherent}. More importantly, the authors in~\cite{Qian2017Nocoherent} demonstrated that even if the time delay is noticeable, it will not affect the signal detection performance.

As $\sigma_{mn} \sim \mathcal{CN}(0,1)$, the noise power equals to 1. Thus, we denote $\alpha^\mathrm{jr} \triangleq P^\mathrm{jr}$ and $\alpha^\mathrm{tr} \triangleq P^\mathrm{tr} $ as the average receive signal-to-noise ratio (SNR) of the direct link from the jammer to the receiver and the average receive SNR of the direct link from the transmitter to the receiver, respectively. Similarly, we denote $\alpha^\mathrm{jb} \triangleq \tilde{\alpha}^\mathrm{j} P^\mathrm{jr}$ and $\alpha^\mathrm{tb} \triangleq \tilde{\alpha}^\mathrm{t} P^\mathrm{tr}$ as the average receive SNR of the backscatter links (jammer-tag-receiver and transmitter-tag-receiver, respectively). Thus, (\ref{eq:receSignal}) can be rewritten as follows:
\begin{equation}
y_{mn} = \underbrace{f_m^\mathrm{t} \sqrt{\alpha^\mathrm{tr}}s_n^\mathrm{t} + f_m^\mathrm{j} \sqrt{\alpha^\mathrm{jr}}s_n^\mathrm{j}}_{\text{direct links}} + \underbrace{f_m^\mathrm{b}  e \Big(g^\mathrm{t} \sqrt{\alpha^\mathrm{tb}}s_n^\mathrm{t} + g^\mathrm{j} \sqrt{\alpha^\mathrm{jb}}s_n^\mathrm{j} \Big)}_{\text{backscatter links}} + \sigma_{mn}.
\end{equation}
As there are $M$ antennas at the receiver, the channel response vector is defined as follows:
\begin{equation}
\mathbf{f}^\mathrm{t} = [f_1^\mathrm{t}, \ldots, f_m^\mathrm{t}, \ldots, f_M^\mathrm{t}]^\mathsf{T},
\end{equation}
\begin{equation}
\mathbf{f}^\mathrm{j} = [f_1^\mathrm{j}, \ldots, f_m^\mathrm{j}, \ldots, f_M^\mathrm{j}]^\mathsf{T},
\end{equation}
\begin{equation}
\mathbf{f}^\mathrm{b} = [f_1^\mathrm{b}, \ldots, f_m^\mathrm{b}, \ldots, f_M^\mathrm{b}]^\mathsf{T}.
\end{equation}
As such, the total signals received at $M$ antennas can be expressed as follows:
\begin{equation}
\mathbf{y}_{n} = \underbrace{\mathbf{f}^\mathrm{t} \sqrt{\alpha^\mathrm{tr}}s_n^\mathrm{t} + \mathbf{f}^\mathrm{j} \sqrt{\alpha^\mathrm{jr}}s_n^\mathrm{j}}_{\text{direct links}} + \underbrace{\mathbf{f}^\mathrm{b}  e \Big(g^\mathrm{t} \sqrt{\alpha^\mathrm{tb}}s_n^\mathrm{t} + g^\mathrm{j} \sqrt{\alpha^\mathrm{jb}}s_n^\mathrm{j} \Big)}_{\text{backscatter links}} + \bm{\sigma}_{n},
\end{equation}
where 
\begin{equation}
\mathbf{y}_{n} = [y_{1n}, \ldots, y_{mn}, \ldots, y_{Mn}]^\mathsf{T},
\end{equation}
\begin{equation}
\bm{\sigma}_{n} = [\sigma_{1n}, \ldots, \sigma_{mn}, \ldots, \sigma_{Mn}]^\mathsf{T}.
\end{equation}
\begin{figure}[!]
	\centering
	\includegraphics[scale=0.3]{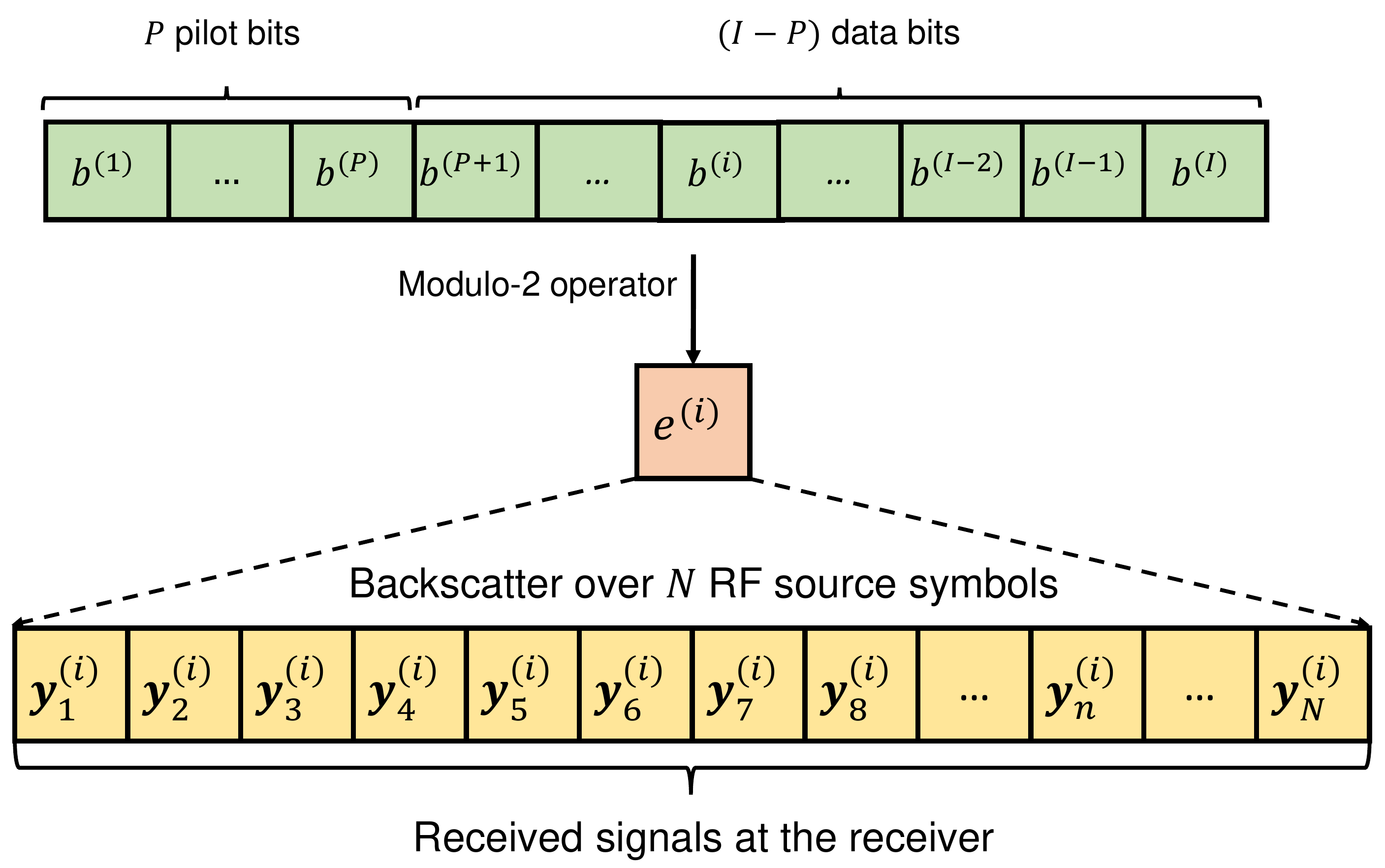}
	\caption{Backscatter frame.}
	\label{fig:backscatter_frame} 
\end{figure}
In this paper, each backscatter frame $\mathbf{b}$ is assumed to contain $I$ bits as follows:
\begin{equation}
\mathbf{b} = [b^{(1)},\ldots, b^{(i)}, \ldots, b^{(I)}],
\end{equation}
where $b^{(i)} \in \{0,1\}, \forall i = 1,2,\ldots, I$. The backscatter frame consists of $P$ pilot bits and $I-P$ data bits as illustrated in Fig.~\ref{fig:backscatter_frame}. The pilot bits are known by the receiver in advance and are used to implicitly capture the channel coefficients for the proposed DL-based detector in Section~\ref{sec:deeplearning}. For simplicity, we assume that the channel remains invariant during one backscatter frame. In this paper, we use the modulo-2 operator to differentially encode the original bits. By using the modulo-2 differential encoder, widely used in ambient backscatter systems~\cite{Wang2016Ambient, Hu2020Ambient, Devineni2019Ambient, Guo2019Nocoherent}, bit ``0'' presents the same state in two consecutive intervals and bit ``1'' corresponds to the transition from non-backscatter to backscatter or from backscatter to non-backscatter~\cite{Wang2016Ambient}. In this way, the receiver can decode the backscattered signals at low BER without requiring long training sequences and CSI according to~\cite{Wang2016Ambient}. Moreover, the modulo-2 differential encoder can also reduce the power consumption of the system~\cite{Wang2016Ambient}. This is very helpful in ambient backscatter systems where the backscatter tags are battery free and transmitting data at low rates. In particular, before being backscattered, each original bit is encoded as follows~\cite{Guo2019Nocoherent}:
\begin{equation}
\label{eq:decode}
e^{(i)} = e^{(i-1)} \oplus b^{(i)},
\end{equation}
where $\oplus$ presents the modulo-2 operator and $\mathbf{e}=[e^{(1)}, \ldots, e^{(i)}, \ldots, e^{(I)}]$ is the encoded symbol vector in which $e^{(0)} = 1$ is the reference symbol. Note that, each backscatter symbol $e^{(i)}$ is backscattered over $N$ RF source symbols as illustrated in Fig.~\ref{fig:backscatter_frame}. This backscattering process can be viewed as a kind of channel coding/modulation. Here, when we increase the value of $N$, the backscatter rate is reduced (i.e., less number of bits can be backscattered over the backscattering channel for a given number of ambient signal symbols), and thus the number of backscattered bits is also reduced. However, with a higher value of $N$, we can achieve a better BER performance as the receiver can receive more backscattered signals to better decode the backscattered bit. This is similar to the case that we use a code word of length $N$ to encode an information bit/symbol.

The received signal during the $i$-th backscatter symbol period can be obtained as follows:
\begin{equation}
\mathbf{y}_{n}^{(i)} = \mathbf{f}^\mathrm{t} \sqrt{\alpha^\mathrm{tr}}{(s_n^\mathrm{t})}^{(i)} + \mathbf{f}^\mathrm{j} \sqrt{\alpha^\mathrm{jr}}{(s_n^\mathrm{j})}^{(i)} + \mathbf{f}^\mathrm{b}  e^{(i)} \Big(g^\mathrm{t} \sqrt{\alpha^\mathrm{tb}}{(s_n^\mathrm{t})}^{(i)} + g^\mathrm{j} \sqrt{\alpha^\mathrm{jb}}{(s_n^\mathrm{j})}^{(i)} \Big) + \bm{\sigma}_{n}^{(i)},
\end{equation}
where $n=1,2,\ldots, N$ and $i=1,2,\ldots, I$. Recall that although various works, e.g., \cite{Darsena2017Modeling, Zhao2018Ambient, Qian2019Achievable} have investigated the maximum achievable rate or capacity of the backscattering channel, to the best of our knowledge, none of them were able to obtain a closed-form expression of the backscattering channel's capacity. In the following, we present a numerical method based on~\cite{Guo2019Nocoherent} to estimate the maximum achievable backscatter rate or the capacity of the ambient backscatter channel. Here we prefer to use the term ``maximum achievable backscatter rate'' that is widely used in the literature of backscattering communications and also more intuitive to readers given the working principle of backscattering communications~\cite{Liu2013Ambient, book}. Note that in Observation~\ref{theo:maxrate}, we obtain the maximum achievable backscatter rate with $N=1$. To the best of our knowledge, all existing works in the literature also cannot obtain the maximum achievable backscatter rate for cases with $N>1$. This could be a potential research direction.

\begin{observation}
\label{theo:maxrate}
Let $I(e;\mathbf{y})$ denote the mutual information between the received signals $\mathbf{y}$ at the receiver and the encoded information $e$. The maximum achievable backscatter rate $R_\text{b}^\dagger$ of the backscatter tag can be numerically obtained as follows:
\begin{equation}
\begin{aligned}
R_\text{b}^\dagger = \max\limits_{\theta_0} I(e;\mathbf{y})&= \max\limits_{\theta_0} \Big( H_\mathrm{b}(\theta_0) - \mathbb{E}_{\{\mathbf{y}_0\}}[H_\mathrm{b}(\omega_0)]\Big)\\
& = \max\limits_{\theta_0} \Big( H_\mathrm{b}(\theta_0) - \int_{\mathbf{y}_0}(\theta_0p(\mathbf{y}_0|e=0) +\theta_1 p(\mathbf{y}_0|e=1))H_\mathrm{b}(\omega_0)d\mathbf{y}_0\Big),
\end{aligned}
\end{equation}
where $\theta_0$ is the probability of the backscattering bit being 0 (the probability of the backscattering bit being 1 is $\theta_1 = 1-\theta_0$), $\mathbf{y}_0$ is a realization of $\mathbf{y}$, $H_\mathrm{b}$ is the binary entropy function, $\omega_0$ is the posterior probability of receiving bit, and $p(\mathbf{y}_0|e=0)$ and $p(\mathbf{y}_0|e=1)$ are the conditional probability density functions corresponding to $e=0$ and $e=1$.
\end{observation}
\begin{proof}
	The proof of Observation~\ref{theo:maxrate} is provided in Appendix~\ref{appendix:maximumbackscatterrate}.
\end{proof}

Denote $\mathbf{Y}^{(i)} = [\mathbf{y}_1^{(i)}, \ldots, \mathbf{y}_n^{(i)}, \ldots, \mathbf{y}_N^{(i)}]^\mathsf{T}$ as the received signal sequence during the $i$-th backscatter symbol period. Next, we develop the ML detector and the DL-based detector to recover the original bit $b^{(i)}$ from $\mathbf{Y}^{(i)}$.

\section{Decoding Backscattered Information at the Receiver with Maximum Likelihood Detector}
\label{sec:decoding}

As mentioned, to the best of our knowledge, there is no solution in the literature that can deal with the considered super-reactive jammer. Our proposed anti-jamming solution can enable the transmitter to lure/trap the jammer to allow the backscatter tag to leverage the RF signals from both the jammer and the transmitter to send the actual information to the receiver. However, decoding the weak backscattered signals is very challenging as the received signals are combinations of the backscattered signals and the direct link signals from the transmitter and the jammer. In this section, we provide details on decoding the backscattered signals by using the optimal ML detector. It is worth noting that the ML detector is well known as an optimal signal detector, and thus it provides an upper-bound performance.

\subsection{Likelihood Functions of Received Signals}

The fundamental task of the ML detector is deriving the likelihood functions of the received signals at the receiver. Specifically, when the backscatter tag backscatters bits ``0'', i.e., $e^{(i)} = 0$, there is no backscattered signals in the received signals. Otherwise, the received signals contain both the direct link signals and the backscattered signals when the tag backscatters bits ``1'', i.e., $e^{(i)} = 1$. Thus, we first derive the channel statistical covariance matrices corresponding to these cases as follows~\cite{Guo2019Nocoherent},~\cite{Zhang2019Constellation}:

\begin{equation}
\begin{aligned}
\mathbf{K}_1 &= (\mathbf{h}_1+\mathbf{h}_2)(\mathbf{h}_1+\mathbf{h}_2)^\text{H} + (\mathbf{h}_3+\mathbf{h}_4)(\mathbf{h}_3+\mathbf{h}_4)^\text{H} + \mathbf{I}_M,\\
\mathbf{K}_0 &= \mathbf{h}_1 \mathbf{h}_1^\text{H} + \mathbf{h}_3 \mathbf{h}_3^\text{H} + \mathbf{I}_M,
\end{aligned}
\end{equation} 
where $\mathbf{h}_1 = \mathbf{f}^\mathrm{t} \sqrt{\alpha^\mathrm{tr}}$, $\mathbf{h}_2 =  g^\mathrm{t} \mathbf{f}^\mathrm{b} \sqrt{\alpha^\mathrm{tb}}$, $\mathbf{h}_3 = \mathbf{f}^\mathrm{j} \sqrt{\alpha^\mathrm{jr}}$, $\mathbf{h}_4 =  g^\mathrm{j} \mathbf{f}^\mathrm{b} \sqrt{\alpha^\mathrm{jb}}$, $\mathbf{I}_M$ is the $M \times M$ identity matrix, and $\mathbf{A}^\text{H}$ is the conjugate transpose of matrix $\mathbf{A}$. $\mathbf{K}_1$ is the channel statistical covariance matrix when backscattering bits ``1'', and $\mathbf{K}_0$ is the channel statistical covariance matrix when backscattering bits ``0''. Note that the channel state information can be obtained through several common estimation techniques that are widely studied in the literature~\cite{Ma2018Blind},~\cite{Zhao2019Channel}. As the noise and the RF signals generated from the transmitter and the jammer follow the CSCG distribution\footnote{Note that our proposed deep learning-based detector does not require this information in advance.}. As such, $\mathbf{y}_n^{(i)}$ is a CSCG distributed vector. Given backscatter symbol $e^{(i)}$ and received signals $\mathbf{y}_n^{(i)}$, we can derive the conditional probability density functions (PDFs) corresponding to $e^{(i)} = 0$ and $e^{(i)} = 1$ as follows:
\begin{equation}
\label{eq:pdf}
\begin{aligned}
p(\mathbf{y}_n^{(i)}|e^{(i)} = 0) &= \frac{1}{\pi^M |K_0|} e^{-{\mathbf{y}_n^{(i)}}^\text{H} K_0^{-1}\mathbf{y}_n^{(i)}}, \\
p(\mathbf{y}_n^{(i)}|e^{(i)} = 1) &= \frac{1}{\pi^M |K_1|} e^{-{\mathbf{y}_n^{(i)}}^\text{H} K_1^{-1}\mathbf{y}_n^{(i)}}.
\end{aligned}
\end{equation}
Based on (\ref{eq:pdf}), the likelihood functions of the received signals sequence $\mathbf{Y}^{(i)} = [\mathbf{y}_1^{(i)}, \ldots, \mathbf{y}_n^{(i)}, \ldots, \mathbf{y}_N^{(i)}]^\mathsf{T}$ given $e^{(i)} = 0$ and $e^{(i)} = 1$ can be expressed as follows~\cite{Guo2019Nocoherent}:
\begin{equation}
\label{eq:likelihoodfunctions}
\begin{aligned}
\mathcal{L}(\mathbf{Y}^{(i)}|e^{(i)} = 0) &= \prod_{n=1}^{N} \frac{1}{\pi^M |K_0|} e^{-{\mathbf{y}_n^{(i)}}^\text{H} K_0^{-1}\mathbf{y}_n^{(i)}},\\
\mathcal{L}(\mathbf{Y}^{(i)}|e^{(i)} = 1) &= \prod_{n=1}^{N} \frac{1}{\pi^M |K_1|} e^{-{\mathbf{y}_n^{(i)}}^\text{H} K_1^{-1}\mathbf{y}_n^{(i)}}.
\end{aligned}
\end{equation}

\subsection{Maximum Likelihood Detector}
Based on (\ref{eq:likelihoodfunctions}), the backscattered symbol $e^{(i)}$ can be obtained through the following ML criterion (i.e., hypothesis)~\cite{Guo2019Nocoherent},~\cite{Zhang2019Constellation}:
\begin{equation}
\hat{e}^{(i)} 	=	\left\{	\begin{array}{ll}
0,	&	\mathcal{L}(\mathbf{Y}^{(i)}|e^{(i)} = 0) > \mathcal{L}(\mathbf{Y}^{(i)}|e^{(i)} = 1),\\
1,	&	\mathcal{L}(\mathbf{Y}^{(i)}|e^{(i)} = 0) < \mathcal{L}(\mathbf{Y}^{(i)}|e^{(i)} = 1),
\end{array}	\right.
\end{equation}
where $\hat{e}^{(i)}$ is the decision result of $e^{(i)}$. Based on (\ref{eq:likelihoodfunctions}), the ML detector for the backscattered symbol $e^{(i)}$ is expressed as follows:
\begin{equation}
\hat{e}^{(i)} 	=	\left\{	\begin{array}{ll}
0,	&	\prod_{n=1}^{N} \frac{1}{\pi^M |K_0|} e^{-{\mathbf{y}_n^{(i)}}^\text{H} K_0^{-1}\mathbf{y}_n^{(i)}} > \prod_{n=1}^{N} \frac{1}{\pi^M |K_1|} e^{-{\mathbf{y}_n^{(i)}}^\text{H} K_1^{-1}\mathbf{y}_n^{(i)}},\\
1,	&	\prod_{n=1}^{N} \frac{1}{\pi^M |K_0|} e^{-{\mathbf{y}_n^{(i)}}^\text{H} K_0^{-1}\mathbf{y}_n^{(i)}} < \prod_{n=1}^{N} \frac{1}{\pi^M |K_1|} e^{-{\mathbf{y}_n^{(i)}}^\text{H} K_1^{-1}\mathbf{y}_n^{(i)}}.
\end{array}	\right.
\end{equation}
Using the logarithm operations, we have
%\begin{equation}
%\sum_{n=1}^{N} {\mathbf{y}_n^{(i)}}^H (K_0^{-1}-K_1^{-1})\mathbf{y}_n^{(i)} \underset{{\hat{e}^{(i)} = 0}}{\overset{\hat{e}^{(i)} = 1}{\gtrless}} N\ln\frac{|K_1|}{|K_0|}.
%\end{equation}
\begin{equation}
\hat{e}^{(i)}	=	\left\{	\begin{array}{ll}
0,	&	\sum_{n=1}^{N} {\mathbf{y}_n^{(i)}}^\text{H} (K_0^{-1}-K_1^{-1})\mathbf{y}_n^{(i)} < N\ln\frac{|K_1|}{|K_0|},\\
1,	&	\sum_{n=1}^{N} {\mathbf{y}_n^{(i)}}^\text{H} (K_0^{-1}-K_1^{-1})\mathbf{y}_n^{(i)} > N\ln\frac{|K_1|}{|K_0|}.
\end{array}	\right.
\end{equation}
Based on $\hat{e}^{(i)}$ we can derive the backscattered bit $e^{(i)}$ and then recover the original bit $b^{(i)}$ based on (\ref{eq:decode}). In this way, the ML detector can achieve the optimal detection performance and can be used as an upper-bound in this work.

%=============================================================
%=============================================================

\section{Deep Learning-Based Signal Detector}
\label{sec:deeplearning}

As mentioned, the performance of conventional signal detection techniques greatly depends on the channel and noise distributions. Specifically, these techniques require the model of the current propagation environment and the channel state information. For example, the optimal ML detector in this paper is formulated for Gaussian-like signals and noise (i.e., CSCG). Nevertheless, if the jammer changes its attack strategy or the system is placed at different wireless environments, these conventional detectors may need to be reformulated based on the environment's parameters to achieve a good performance. To address this problem, in this section, we propose a deep learning (DL) based signal detector to detect backscattered signals at the receiver that can automatically adapt to any channel and noise distributions, and thus broadening the applications of the proposed detector, especially in robust anti-jamming strategies as the reactive jammer may perform different attack strategies. Note that a few deep learning-based detectors have been proposed in the literature~\cite{Fan2019CNN,Ye2018Power,Shea2018Over,Rajendran2018Deep}. However, these detectors may not be feasible for backscattered signals which are weak and contain the mixture versions of the direct link signals from the jammer and the transmitter. Our proposed DL-based signal detector can deal with this problem by using the LSTM network to learn the long-term dependencies of the received signals at the multiple antennas and over $N$ RF source symbols for each backscattered bit. To the best of our knowledge, our proposed detector is the first DL-based detector for ambient backscatter communication systems.

\subsection{Data Preprocessing}
Data processing, as known as data mining, plays a vital role in deep learning. In particular, the quality of the training data and the useful information that can be obtained from it greatly affects the learning performance of the deep learning model. In practice, the raw data will not be clean enough and will not be in a proper format that the deep learning model can easily learn to achieve good results. As such, the data must be preprocessed before feeding to the deep learning model. In the following, we describe in details how the data is processed in this work.

Recall that the received signal during the $i$-th backscatter symbol period is obtained as follows:
\begin{equation}
\mathbf{y}_{n}^{(i)} = \mathbf{f}^\mathrm{t} \sqrt{\alpha^\mathrm{tr}}{(s_n^\mathrm{t})}^{(i)} + \mathbf{f}^\mathrm{j} \sqrt{\alpha^\mathrm{jr}}{(s_n^\mathrm{j})}^{(i)} + \mathbf{f}^\mathrm{b}  e^{(i)} \Big(g^\mathrm{t} \sqrt{\alpha^\mathrm{tb}}{(s_n^\mathrm{t})}^{(i)} + g^\mathrm{j} \sqrt{\alpha^\mathrm{jb}}{(s_n^\mathrm{j})}^{(i)} \Big) + \bm{\sigma}_{n}^{(i)},
\end{equation}
where $n=1,2,\ldots, N$ and $i=1,2,\ldots, I$. Denote $\mathbf{Y}^{(i)} = [\mathbf{y}_1^{(i)}, \ldots, \mathbf{y}_n^{(i)}, \ldots, \mathbf{y}_N^{(i)}]^\mathsf{T}$ as the received signal sequence in the $i$-th symbol period. It is worth noting that the received signal sequence $\mathbf{Y}^{(i)}$ can be shortened to the sample covariance matrix to reduce the size of the training data while maintaining the effectiveness of the deep learning model. Specifically, the sample covariance matrix can be expressed as follows:
\begin{equation}
\mathbf{C}^{(i)} = \frac{1}{N} \sum_{n=1}^{N}\mathbf{y}_{n}^{(i)}{\mathbf{y}_{n}^{(i)}}^\text{H},
\end{equation}
where ${\mathbf{y}_{n}^{(i)}}^\text{H}$ is the conjugate transpose of received signal matrix $\mathbf{y}_{n}^{(i)}$. In this way, the sample covariance matrix $\mathbf{C}^{(i)}$ represents all the useful information of the received signal sequence $\mathbf{Y}^{(i)}$.

\begin{figure}[!]
	\centering
	\includegraphics[scale=0.7]{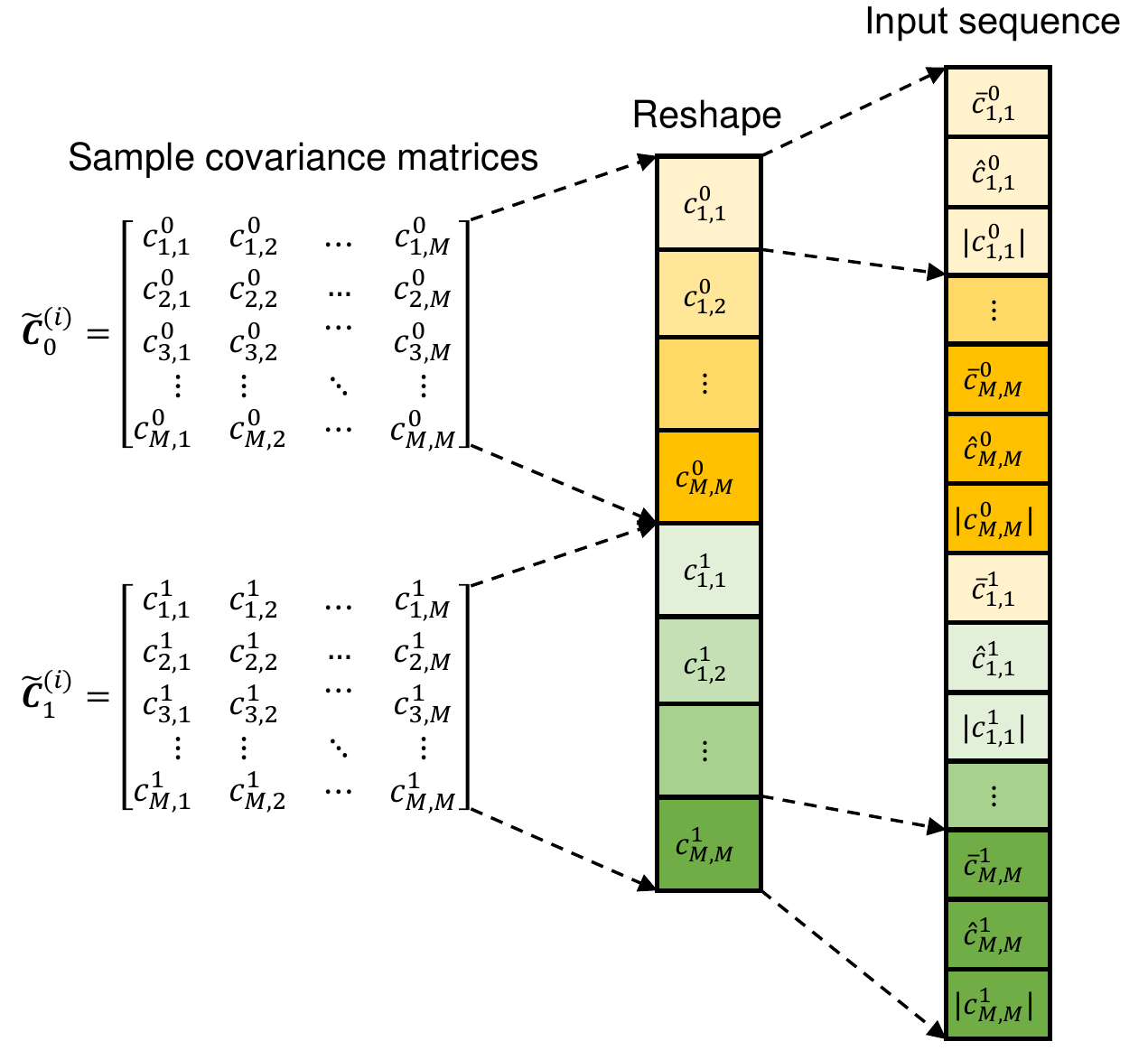}
	\caption{Data preprocessing.}
	\label{fig:reshapingdata}
\end{figure}

It is worth noting that it is impossible to effectively estimate the channel coefficients $\mathbf{h}_1$, $\mathbf{h}_2$, $\mathbf{h}_3$, $\mathbf{h}_4$ in the presence of the jammer, especially when the attack power level is high. Thus, in this paper, instead of explicitly requiring the channel coefficients as in the conventional maximum likelihood detector, our proposed DL-based detector leverages the pilot bits to implicitly capture the channel coefficients. In particular, as mentioned in Section~\ref{sec:received}, each backscatter frame consists of $P$ pilot bits. In this paper, we set the numbers of pilot bits ``0'' and ``1'' equally at $P/2$. As the receiver knows these pilot bits in advance, it can implicitly capture the channel coefficients by observing the received signals corresponding to bit ``0'' and bit ``1''. We denote $\widetilde{\mathbf{K}}_0$ and $\widetilde{\mathbf{K}}_1$ as the average of the estimated covariance matrices of the received signals of pilot bits ``0'' and ``1'', respectively. Then, $\widetilde{\mathbf{K}}_0$ and $\widetilde{\mathbf{K}}_1$ are expressed as follows:
	\begin{equation}
		\begin{aligned}
			\widetilde{\mathbf{K}}_0 = \frac{2}{PN}\sum_{i=1}^{P/2}\sum_{n=1}^{N}\mathbf{y}_n^{(i)}\mathbf{y}_n^{(i)^\mathrm{H}} + \mathbf{I}_M, \mbox{when $e^{(i)} = 0$},\\
			\widetilde{\mathbf{K}}_1 = \frac{2}{PN}\sum_{i=1}^{P/2}\sum_{n=1}^{N}\mathbf{y}_n^{(i)}\mathbf{y}_n^{(i)^\mathrm{H}} + \mathbf{I}_M, \mbox{when $e^{(i)} = 1$}.
		\end{aligned}
\end{equation}

In this way, the receiver does not need to know the channel coefficients to detect the backscattered bits as $\widetilde{\mathbf{K}}_0$ and $\widetilde{\mathbf{K}}_1$ can effectively capture the channel coefficients. Clearly, when the tag backscatters bit ``1'', the elements of $\mathbf{C}^{(i)}$ will be similar with those of $\widetilde{\mathbf{K}}_1$ and different from $\widetilde{\mathbf{K}}_0$. Likewise, when the tag backscatters bit ``0'', the elements of $\mathbf{C}^{(i)}$ will be similar with those of $\widetilde{\mathbf{K}}_0$ and different from $\widetilde{\mathbf{K}}_1$. As such, in~(\ref{eq:estimate}), we multiply $\mathbf{C}^{(i)}$ with the inverse of $\widetilde{\mathbf{K}}_0$ and $\widetilde{\mathbf{K}}_1$ to better represent these similarity and difference. This can result in a high detection accuracy of the deep learning model as demonstrated in Section~\ref{sec:evaluation}.
\begin{equation}
	\label{eq:estimate}
\begin{aligned}
\widetilde{\mathbf{C}}^{(i)}_0 = \mathbf{C}^{(i)}\widetilde{\mathbf{K}}_0^{-1},\\
\widetilde{\mathbf{C}}^{(i)}_1 = \mathbf{C}^{(i)}\widetilde{\mathbf{K}}_1^{-1}.
\end{aligned}
\end{equation}
Finally, both $\widetilde{\mathbf{C}}^{(i)}_0$ and $\widetilde{\mathbf{C}}^{(i)}_1$ are fed to the deep learning model, and then the deep learning model can learn and estimate the backscattered bits.

In Fig.~\ref{fig:reshapingdata}, we show how to reshape $\widetilde{\mathbf{C}}^{(i)}_0$ and $\widetilde{\mathbf{C}}^{(i)}_1$ and feed them to the deep learning model. For simplicity, we remove the indicator $(i)$. $c^0_{i,j}$ with $i,j \in \{1,2,\ldots,M\}$ represents an element of $\widetilde{\mathbf{C}}^{(i)}_0$ and $c^1_{i,j}$ with $i,j \in \{1,2,\ldots,M\}$ represents an element of $\widetilde{\mathbf{C}}^{(i)}_1$. Note that with $M$ antennas at the receive, $\widetilde{\mathbf{C}}^{(i)}_0$ and $\widetilde{\mathbf{C}}^{(i)}_1$ are $M \times M$ matrices. We then reshape $\widetilde{\mathbf{C}}^{(i)}_0$ and $\widetilde{\mathbf{C}}^{(i)}_1$ into two arrays, and each array contains $M \times M$ elements. Thus, we have total $2 \times M \times M$ elements. We then construct the input sequence of the deep learning model by taking the real part $\bar{c}^k_{i,j}$, the imaginary part $\hat{c}^k_{i,j}$, and the absolute value $|c^k_{i,j}|$ of each element $c^k_{i,j}$ with $k \in \{0,1\}$ and $i,j \in \{1,2,\ldots, M\}$. In this way, all aspects of the received signals can be effectively learned by the deep learning model. Then, the input sequence is fed into the deep neural network for learning. As a result, the deep learning model can learn all aspects of the received signals and thus improve the learning efficiency.

\subsection{Deep Neural Network Architecture}
In this work, we adopt Long Short-Term Memory (LSTM)~\cite{LSTM},~\cite{LSTM_MatLab} network that can quickly learn the received signals to detect backscattered bits. An LSTM network is a type of recurrent neural network that can capture and learn the long-term dependencies between input data~\cite{LSTM}. As such, this network architecture is widely adopted in sequence prediction problems. The reasons for using the LSTM network in this work are twofold.

\begin{itemize}
\item First, the input sequence consists of the estimated channel statistical covariances at $M$ antennas (with the same backscattered bit) of the receiver. Thus, with the LSTM network, the dependencies between these covariances are kept track by the memory part of the LSTM unit during the training process. In this way, the LSTM network can detect the backscattered bits more effectively by learning from all received signals at $M$ antennas and the relations between them.
\item Second, as mentioned, the channel remains invariant during one backscatter frame which consists of $I$ bits being backscattered. Each backscattered bit corresponds to an input sequence for the deep learning model to predict. Thus, by using the LSTM network, we can capture and learn the dependencies between these input sequences, and thus keeping track of the channel state during one backscatter frame. As such, the deep learning model can learn the input data more efficiently resulting in a high prediction accuracy.
\end{itemize}

\begin{figure}[!]
	\centering
	\includegraphics[scale=0.27]{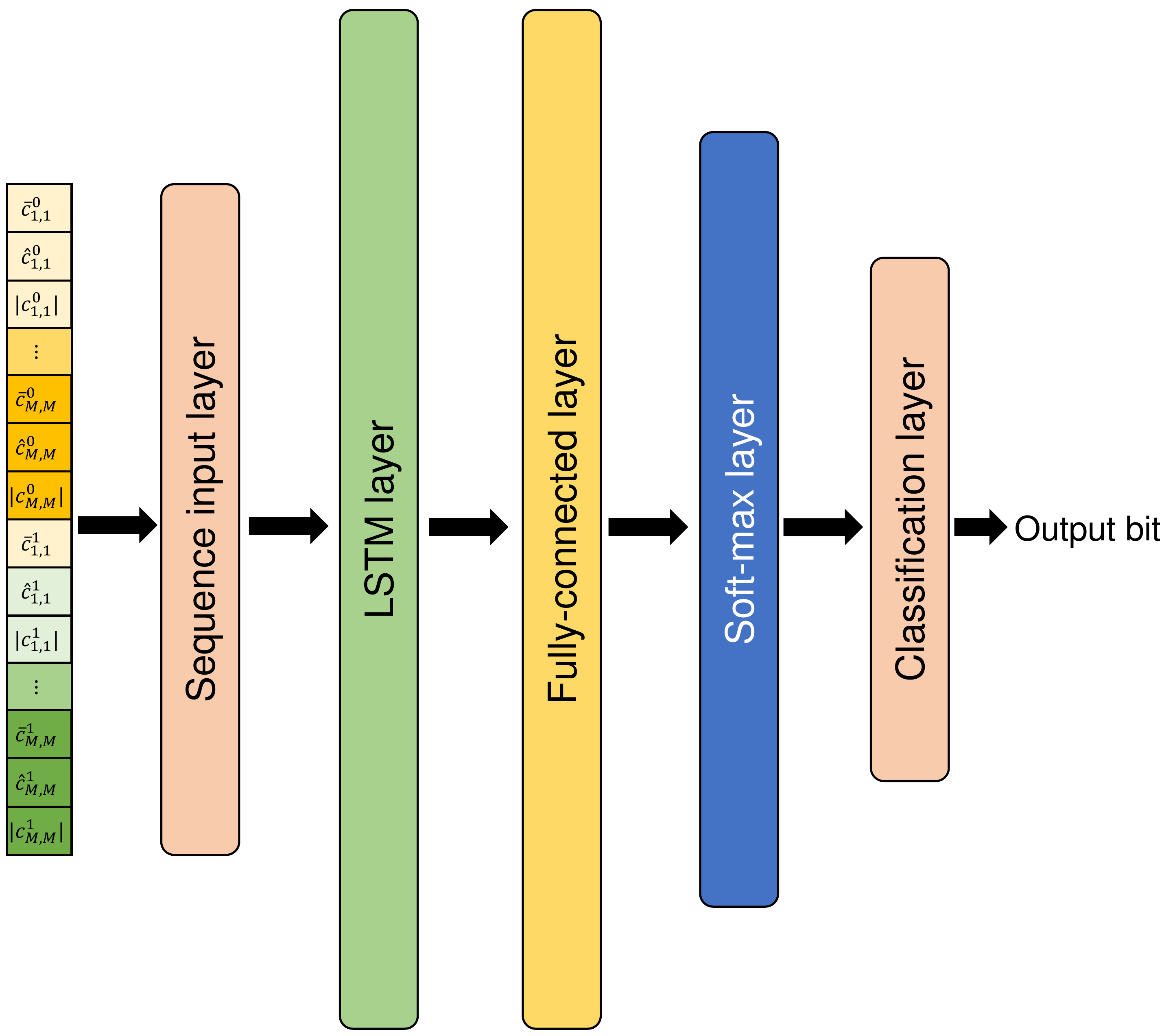}
	\caption{Deep neural network architecture.}
	\label{fig:LSTM}
\end{figure}

In Fig.~\ref{fig:LSTM}, we illustrate the deep neural network architecture in this work. In particular, the deep neural network consists of one sequence input layer, one LSTM layer, one fully connected layer, one softmax layer, and one classification layer. In particular, the sequence input layer is used to import the input data to the deep neural network. After that, the input data is forwarded to the LSTM layer. Specifically, the LSTM layer has an output (also known as the hidden state) denoted by $\mathbf{w}$ and a cell state denoted by $\mathbf{s}$ as illustrated in Fig.~\ref{fig:LSTMBlock}~\cite{LSTM_MatLab}. At each time step $t$, the LSTM layer uses the current output $\mathbf{w}_{t-1}$ and cell state $\mathbf{s}_{t-1}$ as well as the input sequence data to calculate output $\mathbf{w}_t$ and cell state $\mathbf{s}_t$.

\begin{figure}[!]
	\centering
	\includegraphics[scale=0.3]{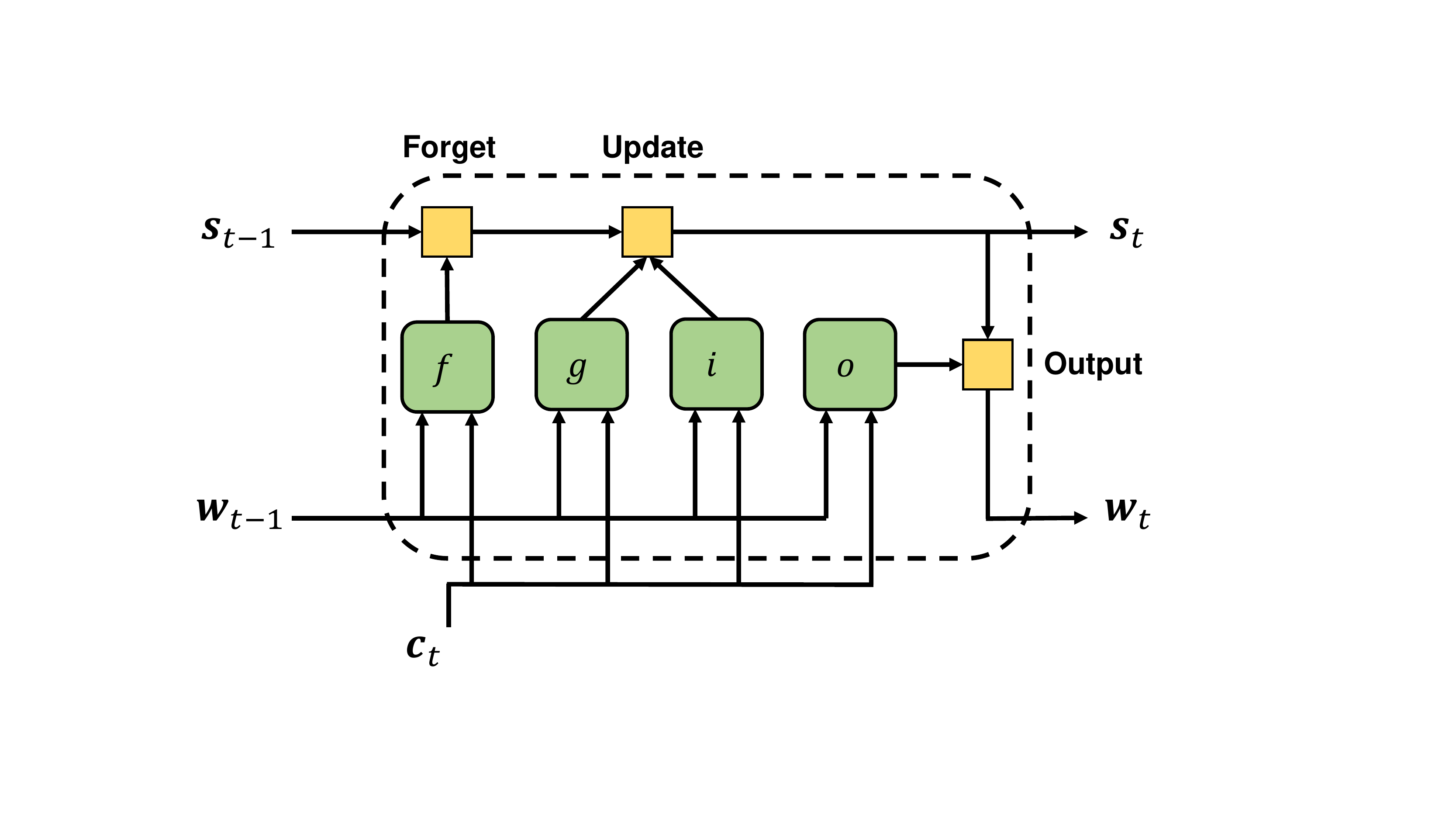}
	\caption{Flow of data at the LSTM layer~\cite{LSTM_MatLab}.}
	\label{fig:LSTMBlock}
\end{figure}

Cell state $\mathbf{s}_t$ represents the information learned from the previous time steps and output state $\mathbf{w}_t$ contains the output for the current time step. In the training process, the LSTM layer updates its weights by using gates. There are three types of weights in the LSTM layer including the input weights $\mathbf{W}$, the recurrent weights $\mathbf{R}$, and the bias $\mathbf{b}$. These matrices are concatenated  from the weights of each gates in the LSTM layers as follows:

\begin{equation}
\mathbf{W}=
\begin{bmatrix}
W_i\\
W_f\\
W_g\\
W_o
\end{bmatrix},
\mathbf{R}=
\begin{bmatrix}
R_i\\
R_f\\
R_g\\
R_o
\end{bmatrix},
\mathbf{b}=
\begin{bmatrix}
b_i\\
b_f\\
b_g\\
b_o
\end{bmatrix},
\end{equation} 
where $i$, $f$, $g$, $o$ represent the input gate, the forget gate, the cell candidate, and the output gate, respectively. In particular, the input state is used to decide what information we will update. At time step $t$, the input gate can be formulated as follows:
\begin{equation}
i_t = \phi_g(W_i\mathbf{c}_t + R_i\mathbf{w}_{t-1}+b_i),
\end{equation}
where $\phi_g$ is the gate activation function. In this work, we use the sigmoid function which is a common activation function in deep learning~\cite{DeepLearningBook}. $\phi_g$ is given by $\phi_g(x) = (1+e^{-x})^{-1}$. $\mathbf{c}_t$ is the input sequence at time step $t$. The forget state is used to decide what information we will remove from the cell state, i.e., remove the long-term dependencies. The forget state can be formulated as follows:
\begin{equation}
f_t = \phi_g(W_f\mathbf{c}_t + R_f\mathbf{w}_{t-1}+b_f).
\end{equation}
The cell candidate is used to add information to the cell state and is formulated as follows:
\begin{equation}
g_t = \phi_c(W_g\mathbf{c}_t + R_g\mathbf{w}_{t-1}+b_g),
\end{equation}
where $\phi_c$ is the state activation function. In this work, we use the hyperbolic tangent function (tanh)~\cite{DeepLearningBook}, which is commonly used to express the cell state of LSTM networks~\cite{LSTM_MatLab}, to derive the state activation function. Finally, the output gate is used control the output of the LSTM layer. This gate can be formulated as follows:
\begin{equation}
o_t = \phi_g(W_o\mathbf{c}_t + R_o\mathbf{w}_{t-1}+b_o).
\end{equation}

By using these gates, the LSTM can update the cell state and the hidden state (i.e., output). In particular, the cell state can be derived as follows:
\begin{equation}
\mathbf{s}_t = f_t \odot \mathbf{s}_{t-1} + i_t \odot g_t,
\end{equation}
where $\odot$ denotes the element-wise multiplication of vectors. Similarly, the hidden state can be expressed as follows:
\begin{equation}
\mathbf{w}_t = o_t \odot \phi_c(\mathbf{s}_t).
\end{equation}

By recurrently updating the cell state and the hidden state, the LSTM layer can learn useful information from the input data. We then use the fully connected hidden layer and the softmax layer between the LSTM and the output layer to transform the output of the LSTM layer to the form that the output layer can use to predict backscattered bits, i.e., bits ``0'' or bits ``1''. In particular, the fully connected hidden layer multiplies its input by a weight matrix and then adds a bias vector. In this work, the weights are initialized with the Glorot initializer~\cite{Glorot2010Xavier}, and the bias is initialized with zeros. Then, the softmax layer applies the softmax function~\cite{DeepLearningBook} to the output of the fully connected layer. Finally, the classification layer calculates the cross entropy loss for the output of the softmax layer. In this way, the receiver can predict the backscattered bits sent from the transmitter. In this paper, we assume that the receiver has sufficient computational power to perform signal processing operations with multiple antennas. It is worth noting that in cases if the receiver does not have sufficient computing resources, the proposed DL-based signal detector can be offloaded to the cloud server and then the trained model will be sent to the receiver to detect the backscattered signals. Moreover, transfer learning~\cite{TransferLearning} can be adopted to significantly reduce the structure of the deep neural network and the training time.

%=============================================================
%=============================================================
\section{Analytical and Simulation Results}
\label{sec:evaluation}
%====================================================================================
%====================================================================================
\subsection{Simulation Setup}

In this work, the Rayleigh fading follows the CSCG distribution with zero mean and unit variance~\cite{Zhang2019Constellation}. Note that the average signal-to-noise ratios of the direct links from the jammer $\alpha^\mathrm{jr}$ and from the transmitter $\alpha^\mathrm{tr}$ to the receiver greatly depend on the environment parameters such as path loss, distances between the jammer and the transmitter to the receiver, transmit power, and antenna gain. For the sake of simplicity, we set $\alpha^\mathrm{jr}$ at 7dB and $\alpha^\mathrm{tr}$ at $5$dB. Note that these parameters are then varied from 1dB to 10dB~\cite{10WJammer} to observe the system performance in different scenarios. Moreover, in ambient backscatter communications, the backscatter link is usually weak~\cite{Zhang2019Constellation}. Thus, we set $\tilde{\alpha}^\mathrm{j} = \tilde{\alpha}^\mathrm{t} = -15$dB. Moreover, to evaluate the performance of our framework in different settings, $M$ is varied from $1$ to $10$. Unless otherwise stated, we set $N=50$, $I=100$, and $P=20$. For the DL-based signal detector, the training set is generated by running $10^5$ Monte Carlo runs. After finishing the training process, the trained model is used to obtain the BER performance through $10^6$ Monte Carlo runs for reliable results. The deep neural network greatly affects the detection performance of the receiver, and thus it needs to be thoughtfully designed. With more layers, the deep neural network can achieve better detection accuracy but requires more time for the learning process. In this work, we design a lightweight deep neural network yet possesses a good detection performance for the receiver. In particular, we employ one sequence input layer, one LSTM layer, one fully connected layer, one softmax layer, and one classification layer as shown in Fig.~\ref{fig:LSTM}. The number of hidden units of the LSTM layer is $1000$. The fully connected layer has the output size of $2$ corresponding to two classes of output bits, i.e., bits ``0'' and bits ``1''. We use Adam optimizer to train the network as it has been demonstrated to be more effective than conventional optimizers~\cite{Adam}. The learning rate is set at $0.001$, which is a common value in literature. It is worth noting that under the attack of a super reactive jammer on the conventional (legitimate) channel, here we focus on evaluating the performance of the backscatter channel from the tag to the receiver.

In this paper, we simulated the system using Matlab. The deep learning model is constructed based on the Deep Learning Toolbox in Matlab. All channels experience Rayleigh fading with with zero mean and unit variance, and the Rayleigh fading generators are independently initialized. In the considered system, the ambient RF source signals (e.g., signals from the jammer and the transmitter in our work) are unknown at the tag and random (e.g., the transmitter actively transmits random/blank signals to attract the jammer and save energy). Thus, we do not need take the legitimate system's information, e.g., modulation, transmitted data, and packet size, into account. Instead, the transmitter's signal is assumed to follow  the standard CSCG distribution with zero mean and unit variance. The performance metrics used for evaluation are the maximum achievable backscatter rate and bit-error-rate. Specifically, the maximum achievable backscatter rate is obtained from Observation 1. This metric is used to evaluate the capacity of the system under jamming attacks. The bit-error-rate is calculated by the ratio between the number of error bits, i.e., not successfully detected by the receiver, and the total transmitted bits. This metric is used to evaluate the reliability of the proposed solution under the jammer attacks.

\subsection{Maximum Achievable Backscatter Rate}
\begin{figure}[!]
	\centering
	\includegraphics[scale=0.43]{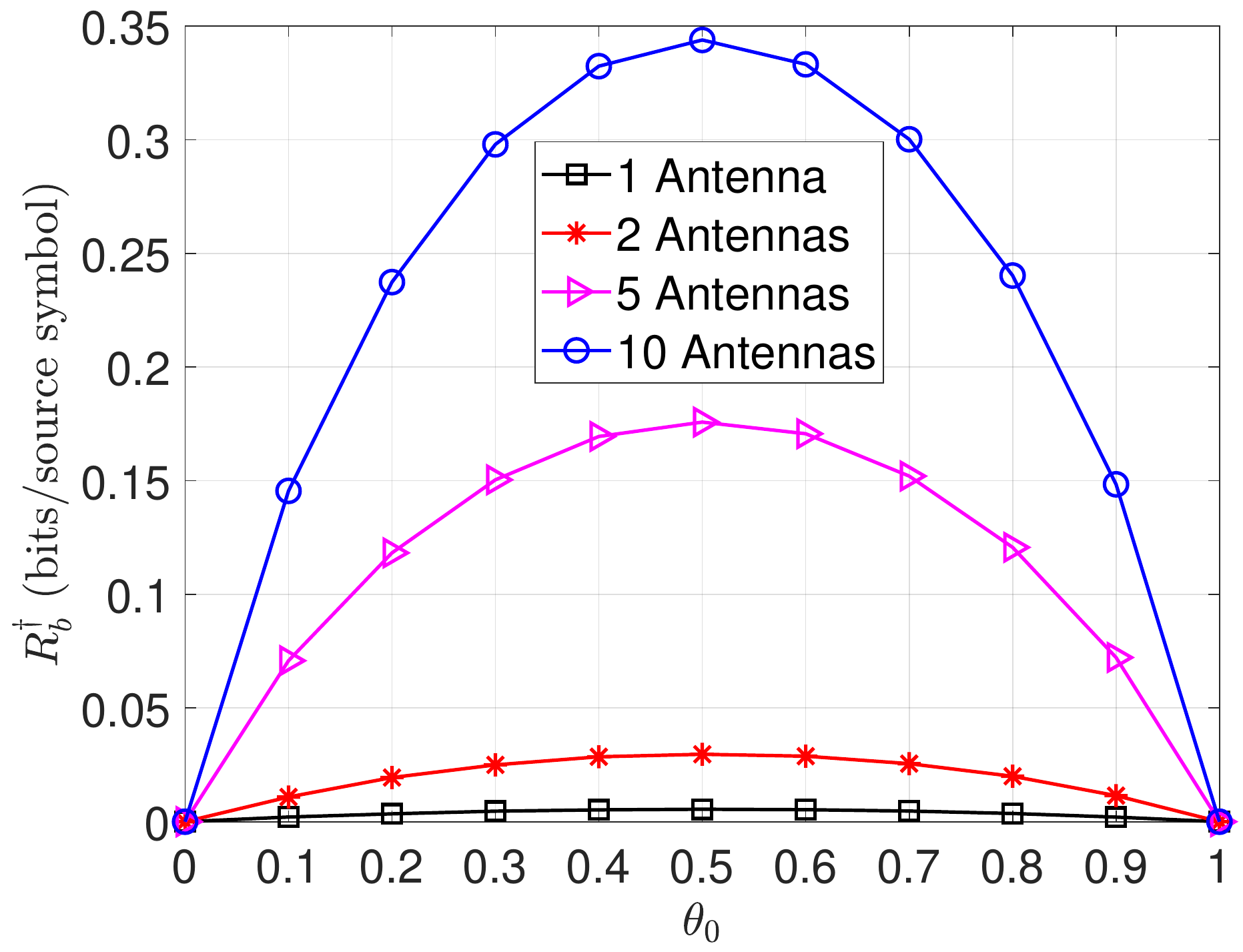}
	\caption{Maximum achievable backscatter rate vs. prior probability when backscattering bits ``$0$''.}
	\label{Fig.varyTheta0}
\end{figure}

We first vary the prior probability of backscattering bits ``$0$'' and observe the maximum achievable backscatter rate of the tag with different number of antennas at the receiver as illustrated in Fig.~\ref{Fig.varyTheta0}. In particular, the maximum achievable backscatter rate is obtained based on Observation~\ref{theo:maxrate} through $10^6$ Monte Carlo runs. It can be observed that the backscatter rate increases with the number of antennas at the receiver. The reason is that with multiple antennas, the receiver can leverage the antenna gain to eliminate the effects of the fading and the direct link interference. As a result, the backscattered signals receiver at the receiver can be enhanced under reactive jamming attacks. It is worth noting that when the probability of backscattering bits ``$0$'' equals 0.5, the backscatter rate is maximized. Thus, in this paper we set $\theta_0 = \theta_1 = 0.5$.

\begin{figure*}[h]
	\centering
	\begin{subfigure}[b]{0.45\textwidth}
		\centering
		\includegraphics[scale=0.4]{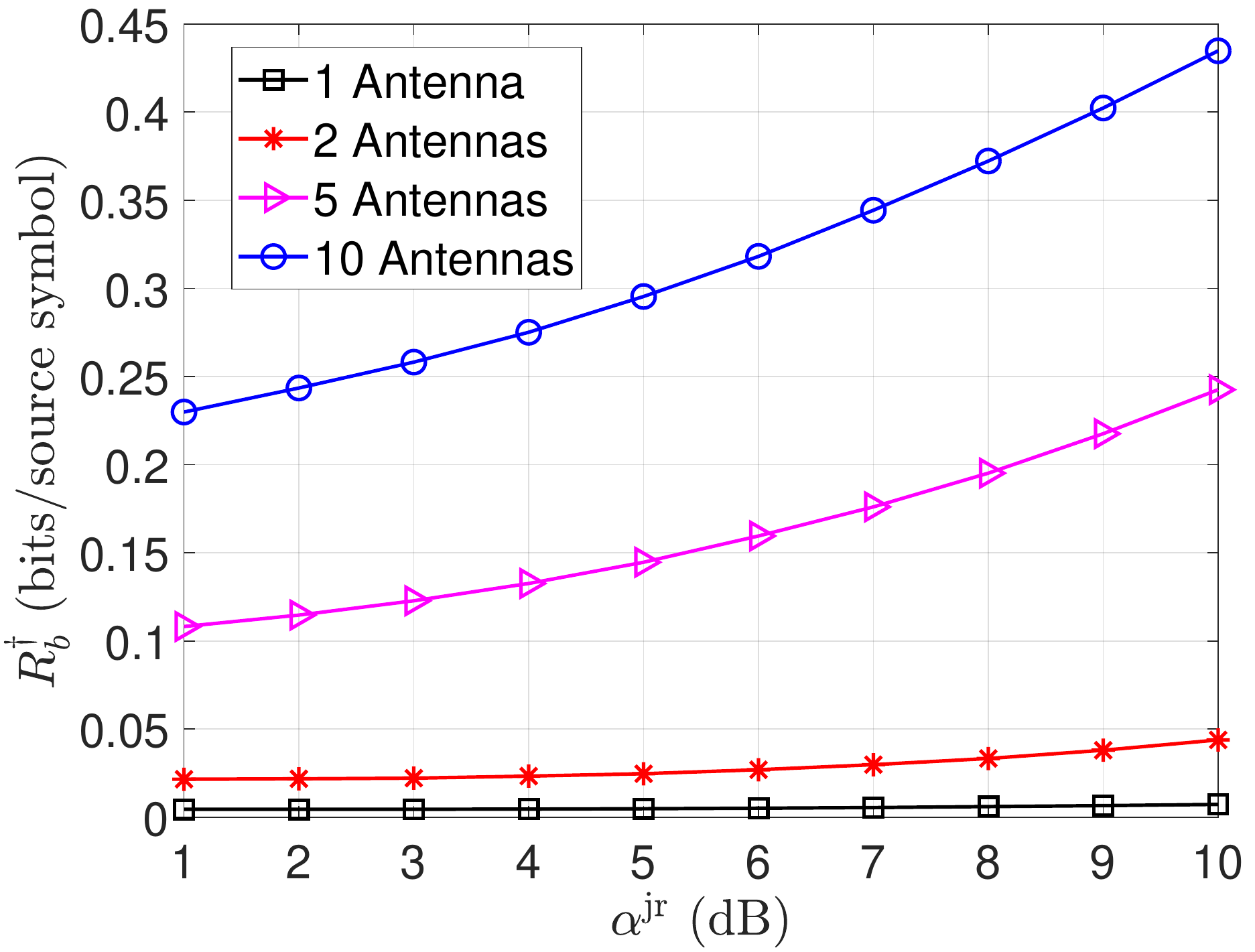}
		\caption{$\alpha^\mathrm{tr} = 5$ dB, $\tilde{\alpha}^\mathrm{j}=\tilde{\alpha}^\mathrm{t}= -15$ dB.}
	\end{subfigure}%
	~ 
	\begin{subfigure}[b]{0.45\textwidth}
		\centering
		\includegraphics[scale=0.4]{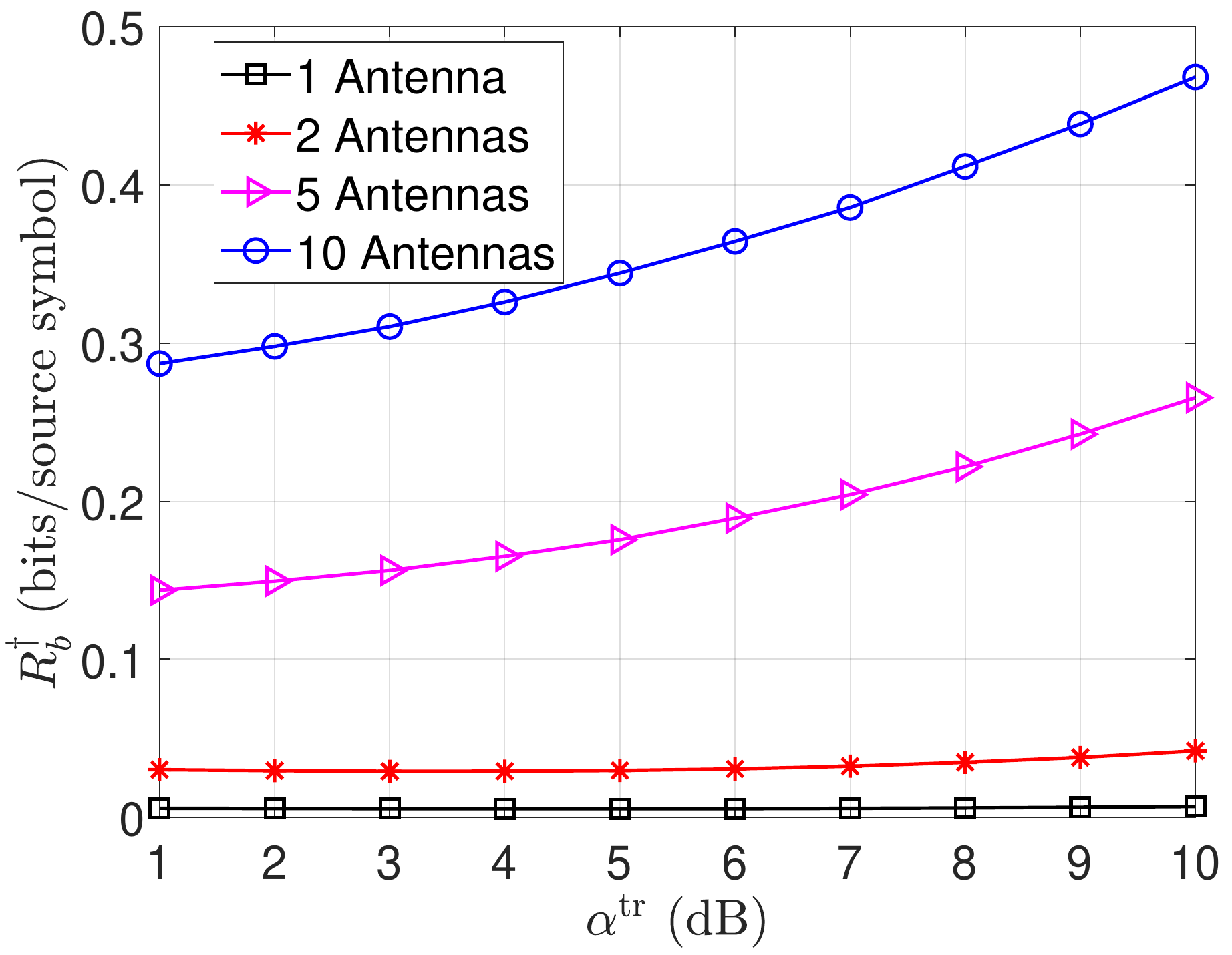}
		\caption{$\alpha^\mathrm{jr} = 7$ dB, $\tilde{\alpha}^\mathrm{j}=\tilde{\alpha}^\mathrm{t}= -15$ dB.}
	\end{subfigure}%
	\caption{Maximum average achievable backscatter rate vs. the average SNRs of the direct link from (a) the jammer and (b) the transmitter.} 
	\label{Fig.PattackTransmitDirectLink}
\end{figure*}

Next, we vary the average SNRs of the direct link from the jammer and from the transmitter to the receiver and evaluate the maximum achievable backscatter rate in various settings as shown in Fig.~\ref{Fig.PattackTransmitDirectLink}(a) and Fig.~\ref{Fig.PattackTransmitDirectLink}(b), respectively. Specifically, in Fig.~\ref{Fig.PattackTransmitDirectLink}(a), the SNR of the direct link from the jammer is varied from $1$ dB to $10$ dB while fixing the SNR of the direct from the transmitter at $5$ dB. It can be observed that the backscatter rate increases with the SNR of the direct link from the jammer. In other words, with our proposed approach, the system can improve the average throughput when the jammer attacks the channel at higher power levels. This is an important finding as our proposed solution can effectively deal with jamming attacks. As we analyzed in Sections~\ref{sec:relatedwork} and~\ref{Sec.System}, conventional methods cannot deal with the considered jamming attacks as the super-reactive jammer can simultaneously attack and listen to the activities of the transmitter. It is worth noting that the higher number of antennas at the receiver, the higher backscatter rate we can achieve. This is due to the fact that the receiver can effectively deal with the direct link interference from the transmitter and the jammer by using multiple antennas. Similarly, in Fig.~\ref{Fig.PattackTransmitDirectLink}(b), we fix the SNR of the direct link from the jammer at $7$ dB and vary the SNR of the direct link from the transmitter from $1$ dB to $10$ dB. It can be observed that, when the transmitter transmits at higher power levels, the tag can backscatter more information to the receiver. Thus, to achieve better performance under jamming attacks, one feasible solution is increasing the transmit power of the transmitter. Note that with two antennas implemented at the receiver, the backscatter rate is very low due to the fading and direct link interference. Again, these negative effects can be mitigated by using more antennas at the receiver.

\begin{figure*}[h]
	\centering
	\begin{subfigure}[b]{0.45\textwidth}
		\centering
		\includegraphics[scale=0.4]{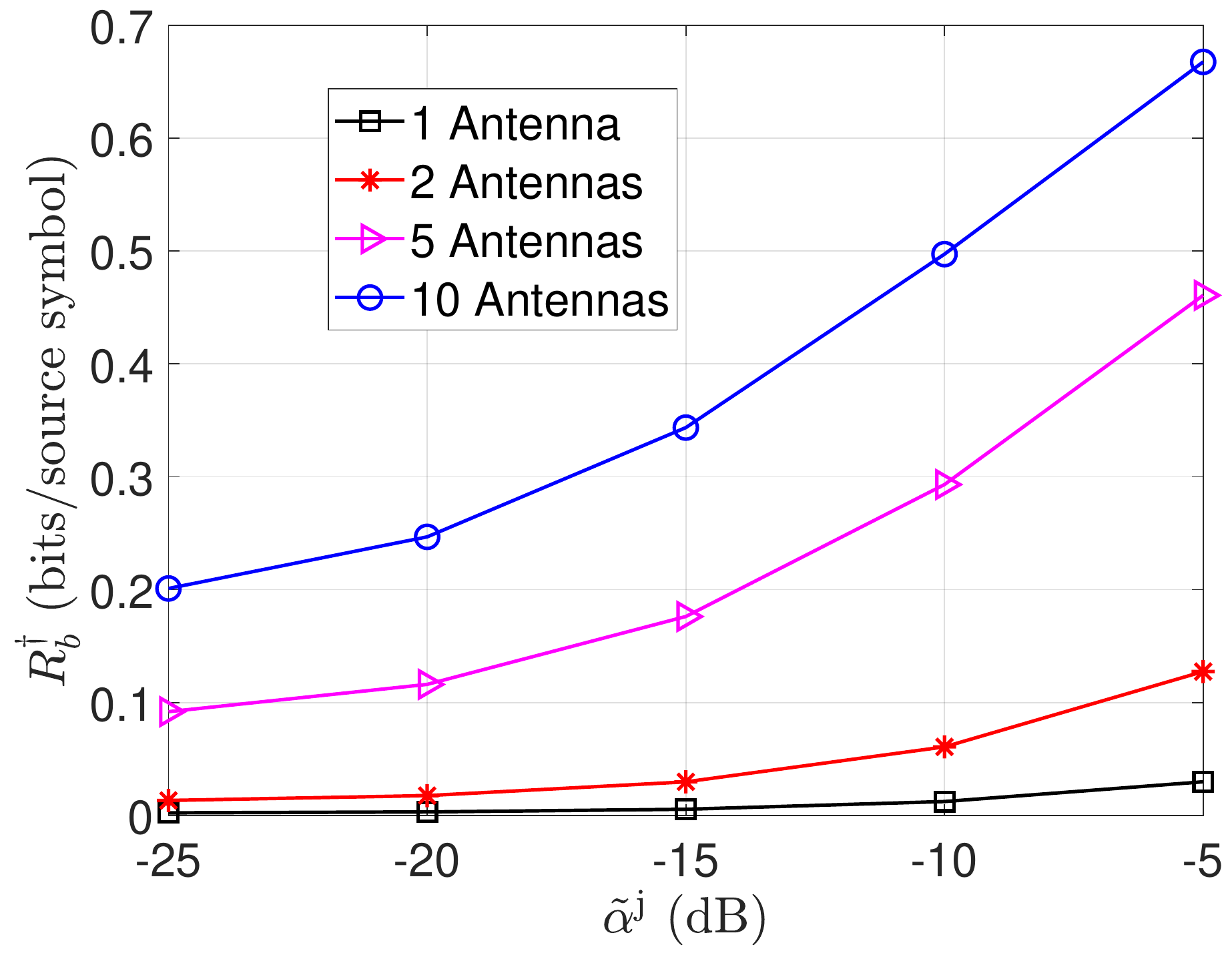}
		\caption{$\alpha^\mathrm{tr} = 5$ dB, $\alpha^\mathrm{jr} = 7$ dB, $\tilde{\alpha}^\mathrm{t}= -15$ dB.}
	\end{subfigure}%
	~ 
	\begin{subfigure}[b]{0.45\textwidth}
		\centering
		\includegraphics[scale=0.4]{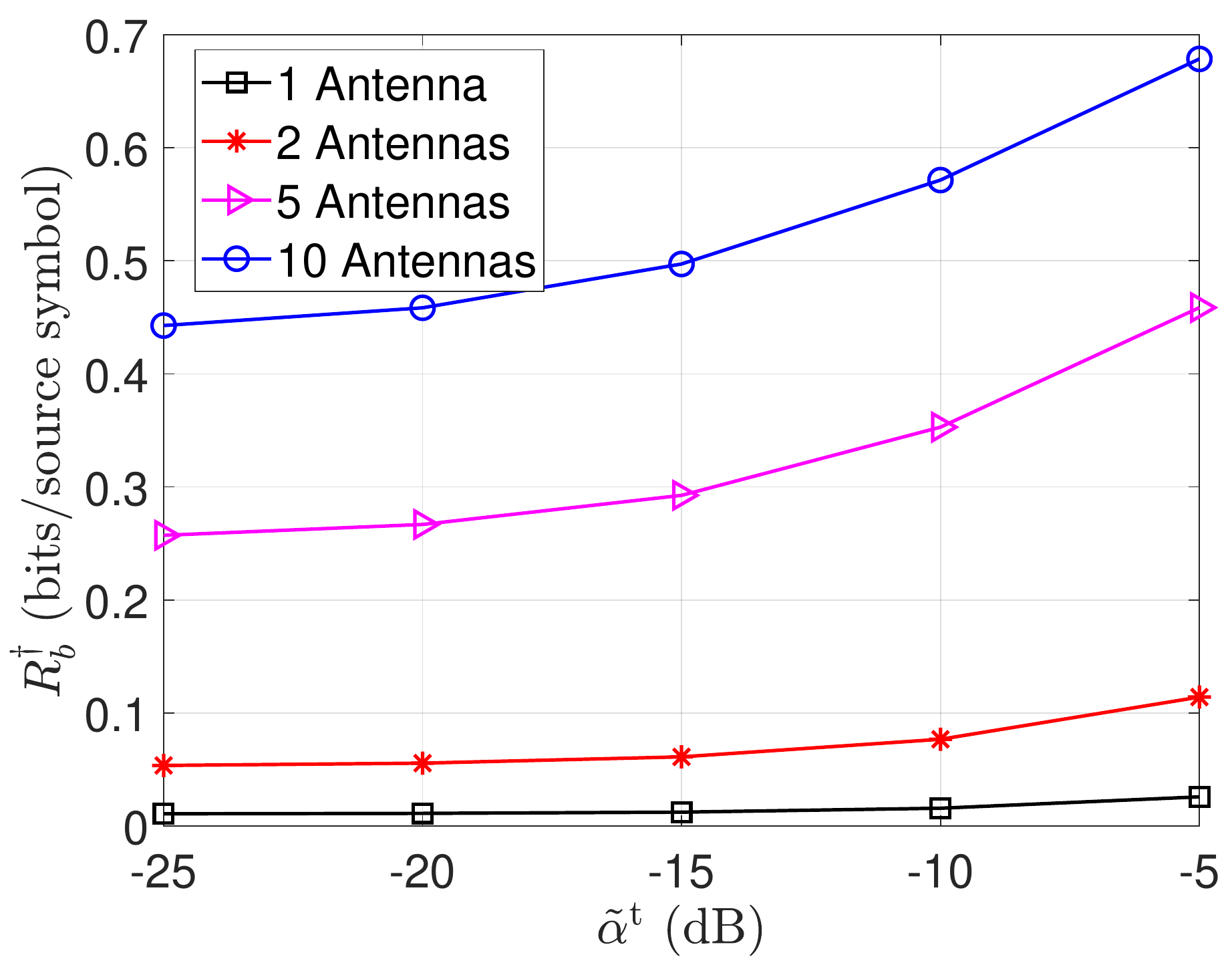}
		\caption{$\alpha^\mathrm{tr} = 5$ dB, $\alpha^\mathrm{jr} = 7$ dB, $\tilde{\alpha}^\mathrm{j}=-10$ dB.}
	\end{subfigure}%
	\caption{Maximum average achievable backscatter rate vs. the SNRs of the backscatter links from (a) the jammer and (b) the transmitter.} 
	\label{Fig.backscatter}
\end{figure*}

Next, we vary the relative SNRs (corresponding to the reflection coefficient, antenna gain, path loss, and tag-to-receiver distance as expressed in (\ref{eq:receSignal})) of the backscatter links from the jammer and the transmitter and observe the maximum achievable backscatter rate of the tag as shown in Fig.~\ref{Fig.backscatter}(a) and Fig.~\ref{Fig.backscatter}(b), respectively. It can be observed that, when the relative SNRs increases, the maximum achievable backscatter rate of the tag also increases. This is because with higher SNRs, the backscattered signals received at the receiver are improved, and thus the receiver can efficiently extract the original information. Again, the receiver can deal with the fading and direct link interference from the transmitter and the jammer to increase the backscatter rate by employing more antennas. Note that in above results, we theoretically obtain the maximum achievable backscatter rate of the backscatter tag which does not require to detect backscattered bits at the receiver. Thus, deep learning is not applied. In the following, we will demonstrate the effectiveness of the DL-based signal detector in terms of BER.

\begin{figure*}[h]
	\centering
	\begin{subfigure}[b]{0.45\textwidth}
		\centering
		\includegraphics[scale=0.4]{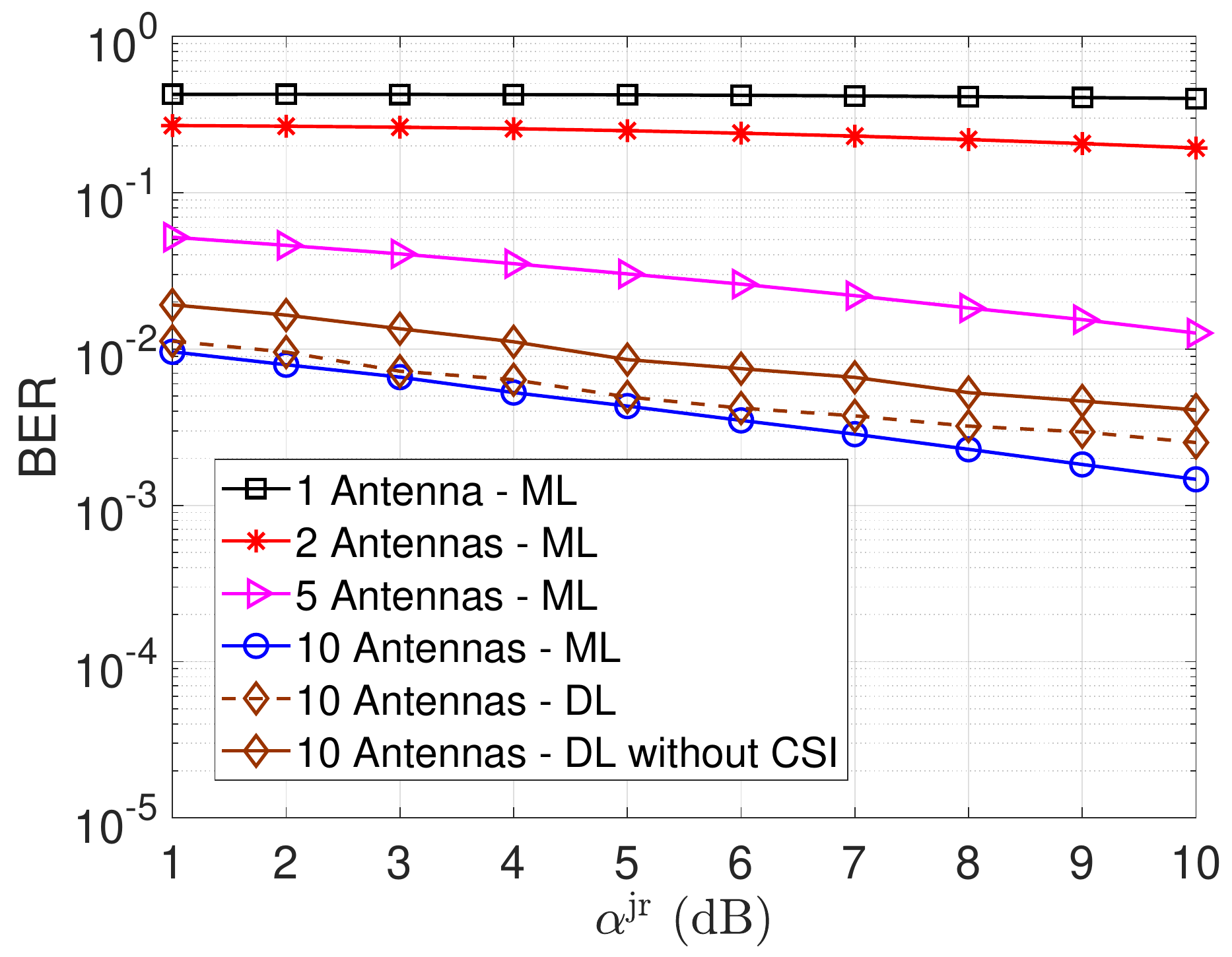}
		\caption{$\alpha^\mathrm{tr} = 5$ dB, $\tilde{\alpha}^\mathrm{j}=\tilde{\alpha}^\mathrm{t}= -15$ dB.}
	\end{subfigure}%
	~ 
	\begin{subfigure}[b]{0.45\textwidth}
		\centering
		\includegraphics[scale=0.4]{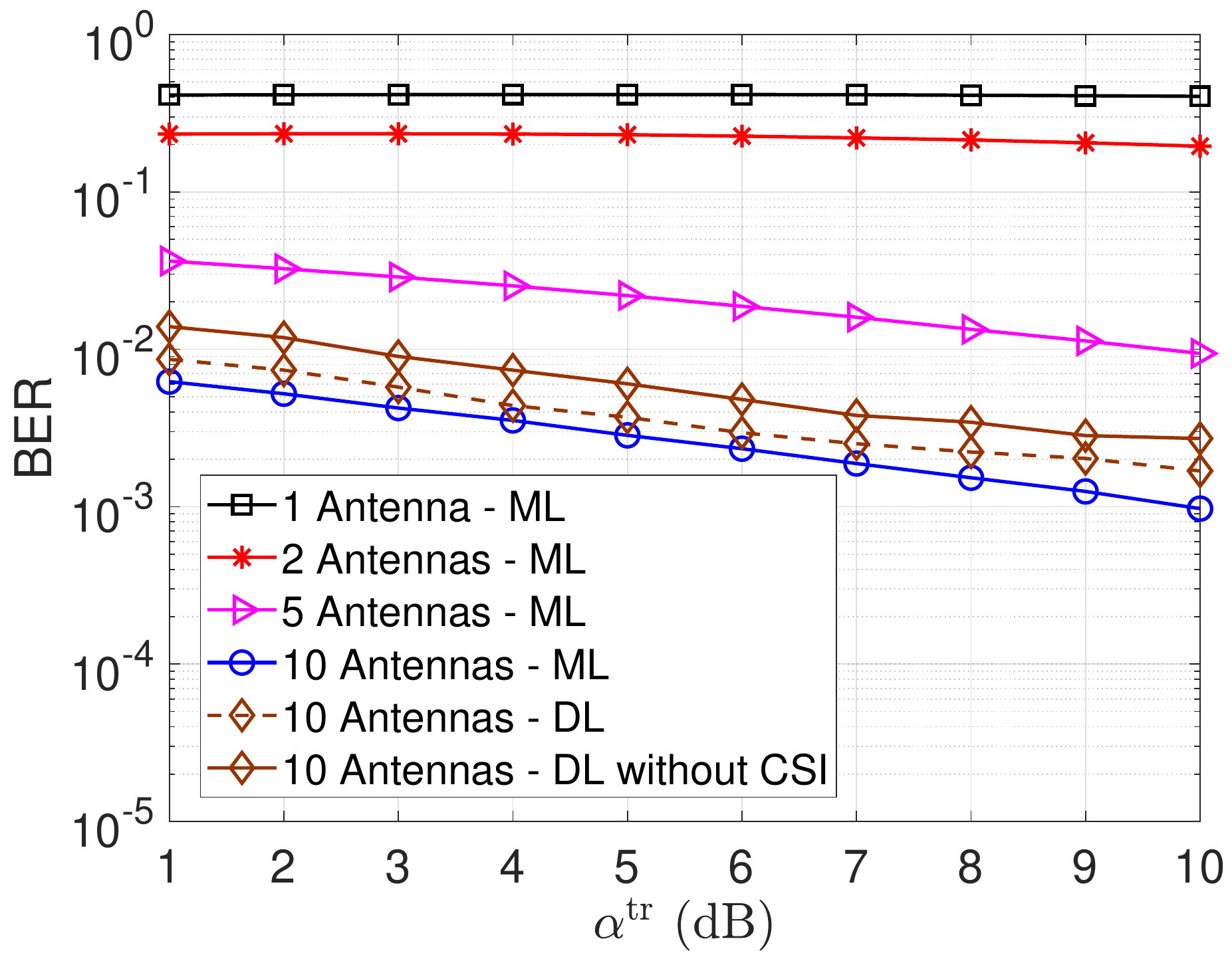}
		\caption{$\alpha^\mathrm{jr} = 7$ dB, $\tilde{\alpha}^\mathrm{j}=\tilde{\alpha}^\mathrm{t}= -15$ dB.}
	\end{subfigure}%
	\caption{BER vs. the average SNRs of the direct link from (a) the jammer and (b) the transmitter.}
	\label{Fig.BERDirectLinks}
\end{figure*}

\subsection{BER Performance}

In this section, the BER of the system is evaluated in different scenarios. We first vary the average SNR of the direct link from the jammer and from the transmitter to the receiver as shown in Fig.~\ref{Fig.BERDirectLinks}(a) and Fig.~\ref{Fig.BERDirectLinks}(b), respectively. In particular, it can be observed from Fig.~\ref{Fig.BERDirectLinks}(a) that the BER reduces when the SNR of the direct link from the jammer to the receiver $\alpha^\mathrm{jr}$ increases. The reason is that with stronger jamming signals, the backscatter tag can backscatter information more effectively, resulting in lower BERs. The similar trend can be observed when increasing the SNR of the direct link from the transmitter to the receiver $\alpha^\mathrm{tr}$ as shown in Fig.~\ref{Fig.BERDirectLinks}(b). As a result, with our proposed anti-jamming mechanism, the higher power the jammer uses to attack the channel, the better BER performance we can achieve. In addition, increasing the transmit power of the transmitter can also improve the BER performance. More importantly, our proposed DL-based detector can achieve the BER performance close to that of the ML detector without requiring the model of the current propagation environment and the channel state information. Again, the BER reduces when we employ more antennas at the receiver. Note that when the SNR increases, the gap between the BERs obtained by our proposed DL-based detector and the optimal ML detector increases. This is because with higher SNRs, the optimal ML detector can detect the backscattered bits based on the likelihood functions of the received signals more efficiently. However, our proposed detector can always achieve the BER close to that of the optimal ML detector without requiring complicated mathematical models and distributions. This trend is the same for the cases with 1, 2 and 5 antennas. However, the DL curves for these cases are not displayed here (for better presentation).

\begin{figure*}[h]
	\centering
	\begin{subfigure}[b]{0.45\textwidth}
		\centering
		\includegraphics[scale=0.4]{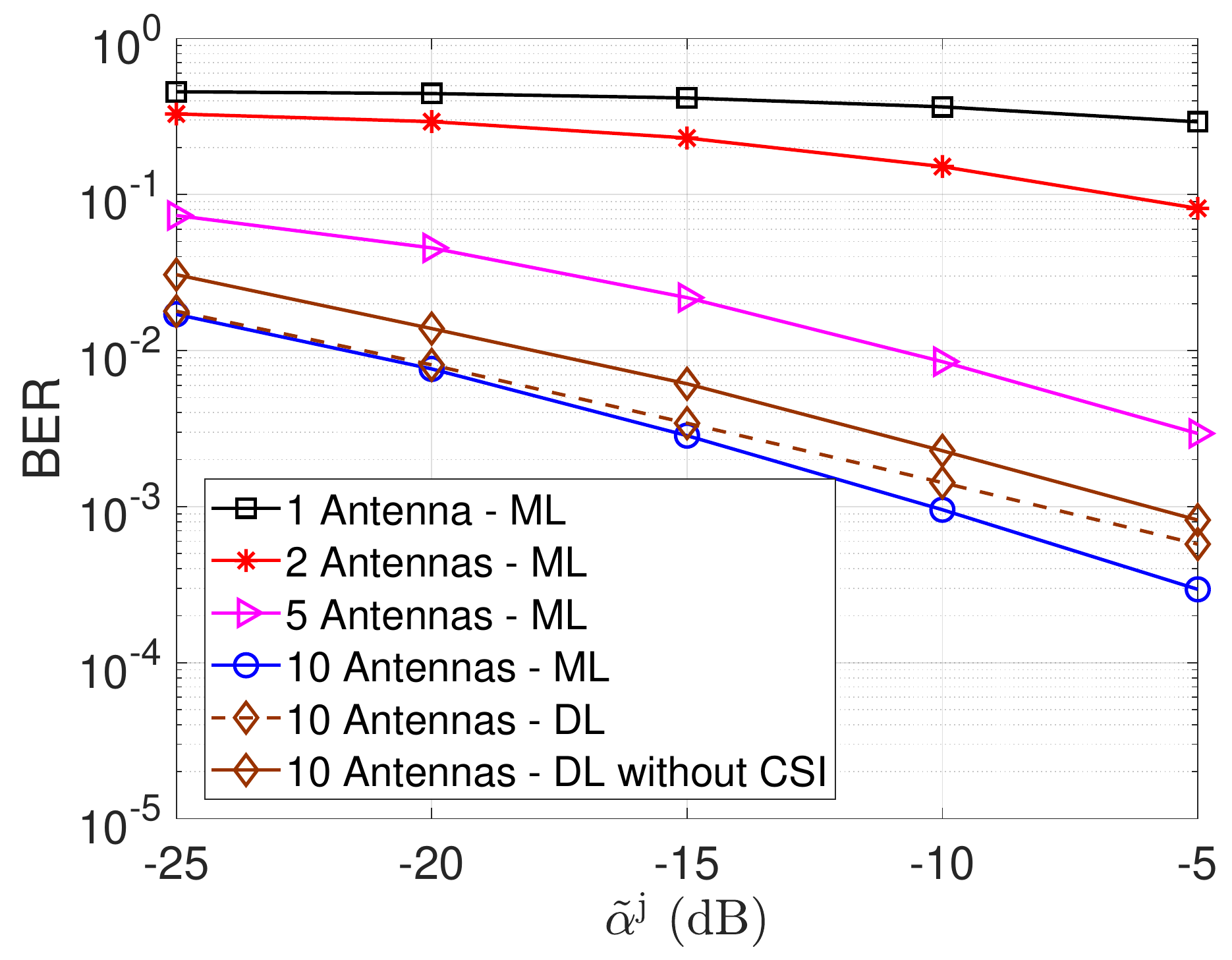}
		\caption{$\alpha^\mathrm{tr} = 5$ dB, $\alpha^\mathrm{jr} = 7$ dB, $\tilde{\alpha}^\mathrm{t}= -15$ dB.}
	\end{subfigure}%
	~
	\begin{subfigure}[b]{0.45\textwidth}
		\centering
		\includegraphics[scale=0.4]{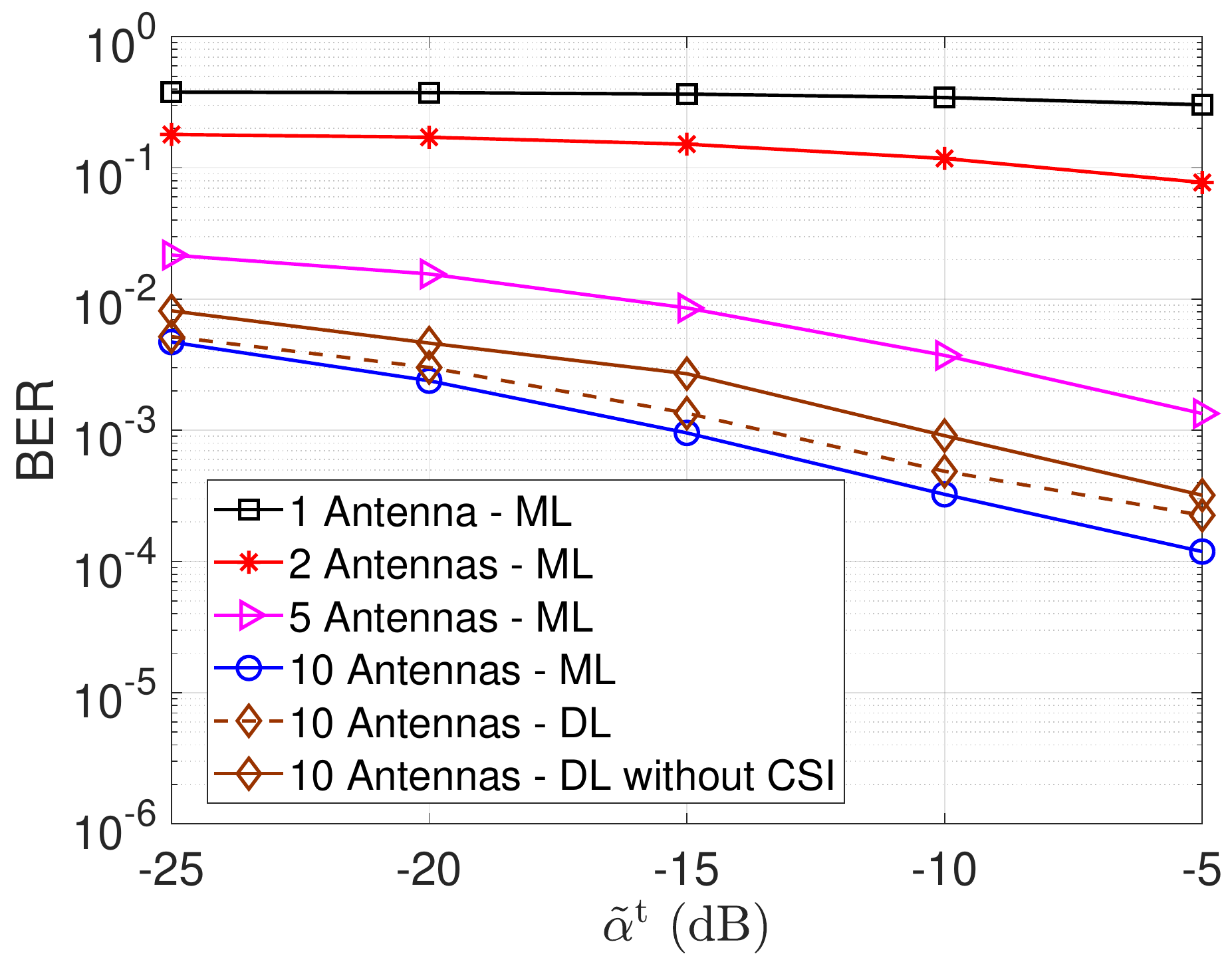}
		\caption{$\alpha^\mathrm{tr} = 5$ dB, $\alpha^\mathrm{jr} = 7$ dB, $\tilde{\alpha}^\mathrm{j}=-10$ dB.}
	\end{subfigure}%
	\caption{BER vs. the SNRs of the backscatter links from (a) the jammer and (b) the transmitter.}
	\label{Fig.BERBackscatterLinks}
\end{figure*}

Next, the SNRs of the backscatter links (i.e., jammer-tag-receiver and transmitter-tag-receiver) are varied to evaluate the BER performance as shown in Fig.~\ref{Fig.BERBackscatterLinks}(a) and Fig.~\ref{Fig.BERBackscatterLinks}(b), respectively. In particular, when $\tilde{\alpha}^\mathrm{j}$ and $\tilde{\alpha}^\mathrm{t}$ increase, the BER is reduced. This is stemmed from the fact that when the SNRs of the backscatter links increase, the received backscattered signals are stronger, and thus the backscatter communications are more reliable. Again, the more antennas we have at the receiver, the better BER performance we can achieve. It is worth noting that our proposed DL-based detector achieves the BER performance close to that of the ML detector.

\begin{figure*}[h]
	\centering
	\begin{subfigure}[b]{0.45\textwidth}
		\centering
		\includegraphics[scale=0.4]{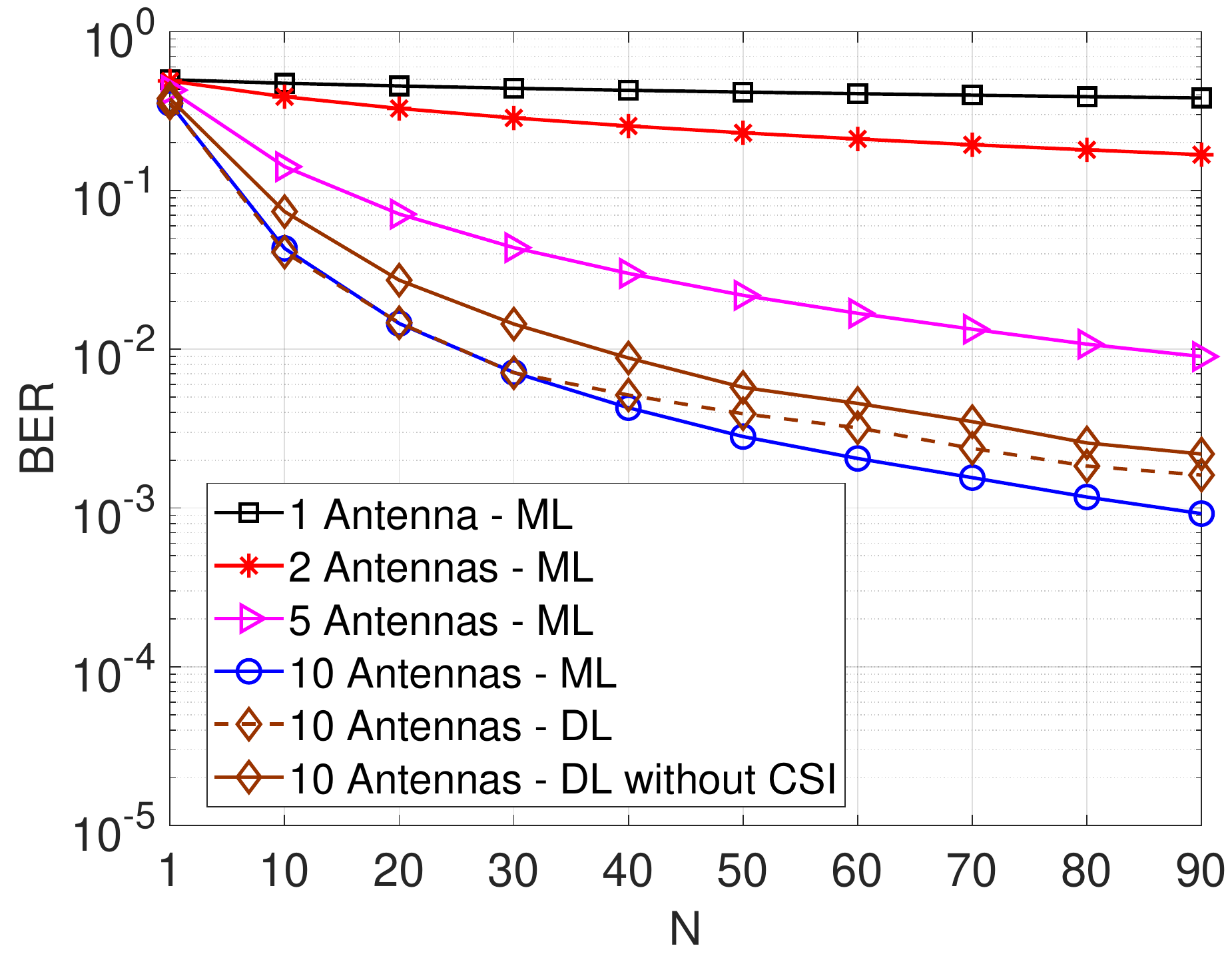}
		\caption{$\alpha^\mathrm{jr} = 7$ dB, $\alpha^\mathrm{tr} = 5$ dB, $\tilde{\alpha}^\mathrm{j}=\tilde{\alpha}^\mathrm{t}= -15$ dB}
	\end{subfigure}%
	~
	\begin{subfigure}[b]{0.45\textwidth}
		\centering
		\includegraphics[scale=0.4]{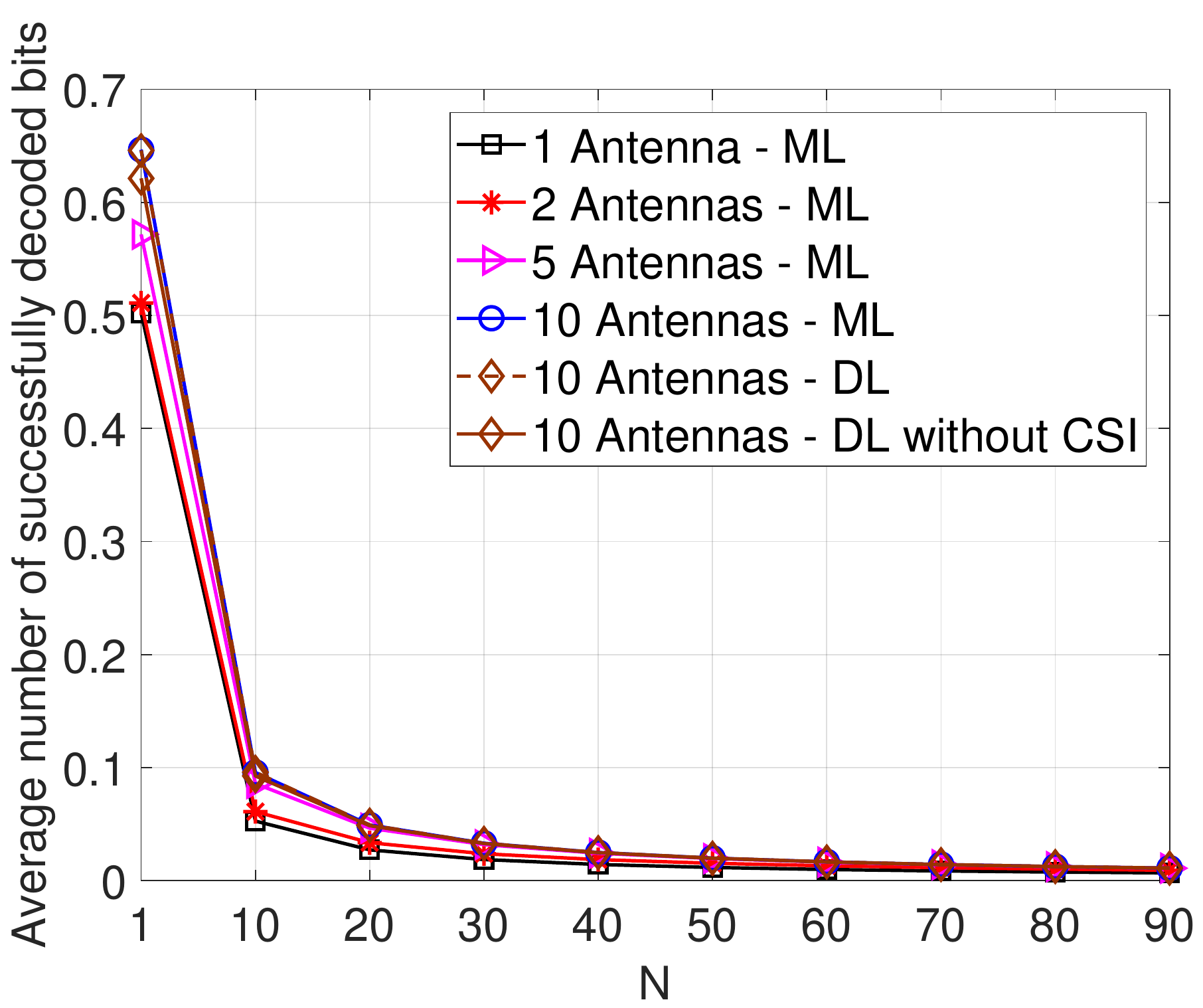}
		\caption{$\alpha^\mathrm{jr} = 7$ dB, $\alpha^\mathrm{tr} = 5$ dB, $\tilde{\alpha}^\mathrm{j}=\tilde{\alpha}^\mathrm{t}= -15$ dB}
	\end{subfigure}%
	\caption{(a) BER and (b) Average number of successfully detected bits vs. $N$.}
	\label{Fig.BERN}
\end{figure*}

Finally, we vary the spreading factor $N$ to evaluate the BER performance of the system. Recall that the backscatter rate should be lower than the data rate of the RF sources to help the receiver decode the backscattered signals more efficiently~\cite{Liu2013Ambient}. Intuitively, if the we backscatter one information symbol/bit per $N$ ambient signal symbols, then the larger $N$ (or the lower backscatter rate), the more reliable decoding of the information bit/symbol can be achieved at the backscattering receiver. It can be observed from Fig.~\ref{Fig.BERN}(a) that the BER reduces when $N$ is increased. In Fig.~\ref{Fig.BERN}(b), we illustrate the average number of successfully detected bits when N is varied from 1 to 100. In particular, the average number of successfully detected bits is the number of bits that are successfully detected at the receiver during each RF source symbol. Clearly, a higher value of $N$ leads to a lower number of successfully detected bits which is the number of bits that are successfully detected at the receiver during each RF source symbol. The trade-off between $N$ and the backscatter rate depends on the system performance objective. From the figure, we can observe that a higher value of $N$ results in a lower backscatter rate but lower BER. As such, the optimal trade-off depends on the system requirements. For example, under our simulation settings, N should be less than 50 if one requires the average number of backscattered bits successfully received at the receiver to be higher than 0.02.

%====================================================================================
%====================================================================================
\section{Conclusion}
\label{sec:conclusion}
In this paper, we have proposed the state-of-the-art framework to deal with the super-reactive jammer that has virtually unlimited power and can simultaneously attack the channel and detect activity of the transmitter on the channel. In particular, we have first introduced the novel deception mechanism that allows the transmitter to attract the jammer and leverage the jamming signals to transmit information to the receiver by using the ambient backscatter tag. To improve the BER performance, we have proposed to use multiple antennas at the receiver. Then, the maximum likelihood detector is developed to detect the backscattered bits at the receiver. Although achieving the optimal performance, the ML solution imposes a high computation complexity and is only applicable to a specific channel distribution. Thus, we have proposed the novel deep learning-based detector with the latest advances in Long Short Team Memory networks, which can efficiently deal with weak backscattered signals under super-reactive jamming attacks. Through the simulation results and theoretical analysis, we have shown that under our proposed framework, the more power the jammer uses to attack the channel, the better performance in terms of BER and throughput the system can achieve. In addition, the proposed deep learning-based signal detector can achieve the BER performance close to that of the optimal ML detector.

%====================================================================================
%====================================================================================
\appendices
\section{The proof of Observation~\ref{theo:maxrate}}
\label{appendix:maximumbackscatterrate}
Similar to~\cite{Guo2019Nocoherent}, in the following, we will show how to obtain the maximum achievable backscatter rate of the backscatter tag. The mutual information $I(e, \mathbf{y})$ can be expressed as follows:
	\begin{equation}
		I(e, \mathbf{y}) = H(e) - H(e|\mathbf{y}) = H_\mathrm{b}(\theta_0) - \mathbb{E}_{\{\mathbf{y}_0\}}[H(e|\mathbf{y}_0)],
	\end{equation}
	where $H(e)$ is the entropy of $e$, $\mathbf{y}_0$ is one realization of $\mathbf{y}$, $H(e|\mathbf{y})$ and $H(e|\mathbf{y}_0)$ are the conditional entropy of $e$ given $\mathbf{y}$ and $\mathbf{y}_0$, respectively. $\theta_0$ is the prior probability of backscattering bit ``0''. $H_\mathrm{b}(\theta_0)$ denotes the binary entropy function of $\theta_0$ and can be expressed as follows:
	\begin{equation}
		\label{eq:Ctheta0}
		H_\mathrm{b}(\theta_0) \triangleq -\theta_0 \log_2 \theta_0 -\theta_1 \log_2 \theta_1,
	\end{equation}
	where $\theta_1$ is the prior probability of backscattering bit ``1'', and $\theta_1 = 1 - \theta_0$.

As $H_\mathrm{b}(\theta_0)$ is independent of all the channel coefficients, $R_\text{b}^\dagger$ can be rewritten as
	\begin{equation}
		\label{eq:max_rate}
		R_\text{b}^\dagger = \max\limits_{\theta_0} \Big(H_\mathrm{b}(\theta_0) - \mathbb{E}_{\{\mathbf{y}_0\}}[H(e|\mathbf{y}_0)]\Big).
	\end{equation}
	It is worth noting that $H(e|\mathbf{y}_0)$ is averaged with respect to $\mathbf{y}_0$ given different prior distributions. In the following, we transform $H(e|\mathbf{y}_0)$ into a binary entropy function. First, we denote $p(e=j|\mathbf{y}_0)$ as the posterior probability of receiving bit $j \in \{0,1\}$ given $\mathbf{y}_0$ as follows:
	\begin{equation}
		p(e=j|\mathbf{y}_0) = \frac{\theta_jp(\mathbf{y}|e=j)}{\theta_0p(\mathbf{y}_0|e=0) +\theta_1 p(\mathbf{y}_0|e=1)},
	\end{equation}
	where $p(\mathbf{y}_0|e=0)$ and $p(\mathbf{y}_0|e=1)$ are the conditional probability density functions corresponding to $e=0$ and $e=1$ as expressed in (28). We then define $\omega_j = p(e=j|\mathbf{y}_0)$ with $j \in \{0,1\}$. Given the above, we can derive $H(e|\mathbf{y}_0)$ as follows:
	\begin{equation}
		\label{eq:H_b}
		H(e|\mathbf{y}_0) = -\sum_{j=0}^{1}\omega_j\log_2\omega_j = H_\mathrm{b}(\omega_0).
	\end{equation}
	Substituting (\ref{eq:H_b}) into (\ref{eq:max_rate}), the maximum achievable backscatter rate can be expressed as follows:
	\begin{equation}
		\label{eq:maximum_backscatter_rate}
		R_\text{b}^\dagger = \max\limits_{\theta_0} \Big(H_\mathrm{b}(\theta_0) - \mathbb{E}_{\{\mathbf{y}_0\}}[H_\mathrm{b}(\omega_0)]\Big) = \max\limits_{\theta_0} \Big( H_\mathrm{b}(\theta_0) - \int_{\mathbf{y}_0}(\theta_0p(\mathbf{y}_0|e=0) +\theta_1 p(\mathbf{y}_0|e=1))H_\mathrm{b}(\omega_0)d\mathbf{y}_0\Big).
\end{equation}

\end{document}